\pdfoutput=1

\documentclass[10pt,journal]{IEEEtran}
\usepackage{cite}
\usepackage{amsmath,amssymb,amsfonts}
\usepackage{graphicx}
\usepackage{textcomp}
\usepackage{xcolor}
\usepackage{mathrsfs}
\usepackage[noend]{algpseudocode}
\usepackage{algorithmicx,algorithm}
\usepackage{epstopdf}
\usepackage{pifont}
\usepackage{amsthm}

\usepackage{bm}
\usepackage{subfigure}
\usepackage{setspace}
\usepackage{verbatim}
\usepackage[bookmarks,colorlinks]{hyperref}
\usepackage{acronym}
\usepackage[marginal]{footmisc}
\def\BibTeX{{\rm B\kern-.05em{\sc i\kern-.025em b}\kern-.08em
    T\kern-.1667em\lower.7ex\hbox{E}\kern-.125emX}}

\acrodef{2d}[2D]{two-dimensional}
\acrodef{crb}[CRB]{Cram$ {\rm \acute{e}} $r-Rao lower bound}
\acrodef{bss}[BSS]{blind source separation}
\acrodef{mmv}[MMV]{multiple measurement vectors}
\acrodef{ica}[ICA]{independent component analysis}
\acrodef{jade}[JADE]{Joint Approximate Diagonalization of Eigen-matrices}
\acrodef{fim}[FIM]{Fisher information matrix}
\acrodef{svd}[SVD]{singular value decomposition}
\acrodef{snr}[SNR]{signal-to-noise rate}
\acrodef{rmse}[RMSE]{root mean square error}
\acrodef{nls}[NLS]{nonlinear least squares}
\acrodef{vna}[VNA]{vector network analyzer}
\acrodef{bf}[BF]{beamforming}
\acrodef{doa}[DoA]{directions of arrival}
\acrodef{mf}[MF]{matched filtering}
\acrodef{admm}[ADMM]{Alternating Direction Method of Multipliers}
\acrodef{music}[MUSIC]{Multiple Signal Classification}
\acrodef{cs}[CS]{compressed sensing}

\newtheorem{prop}{Proposition}

\begin{document}

\title{Direction Finding in Partly Calibrated Arrays Exploiting the Whole Array Aperture}

\author{Guangbin Zhang, Tianyao Huang
, Yimin Liu, 
Xiqin Wang
and Yonina C. Eldar
\thanks{\quad This work
was supported by the National Natural Science Foundation of China under
Grants 62171259. 
G. Zhang, T. Huang, Y. Liu, and X. Wang are with the Department
of Electronic Engineering, Tsinghua University, Beijing 100084, China
(e-mail: zgb18@mails.tsinghua.edu.cn; \{huangtianyao, 
yiminliu,  wangxq\_ee\}@tsinghua.edu.cn). 
Yonina C. Eldar is with the Department of EE Technion, Israel Institute of Technology, Haifa, Israel (e-mail:yonina@ee.technion.ac.il).
T. Huang is the corresponding author.} 
}


\maketitle

\begin{abstract}


We consider the problem of direction finding using partly calibrated arrays, a distributed subarray with position errors between subarrays.
The key challenge is to enhance angular resolution in the presence of position errors.
To achieve this goal, existing algorithms, such as subspace separation and sparse recovery, have to rely on multiple snapshots, which increases the burden of data transmission and the processing delay.
Therefore, we aim to enhance angular resolution using only a single snapshot.
To this end, we exploit the orthogonality of the signals of partly calibrated arrays.
Particularly, we transform the signal model into a special multiple-measurement model, show that there is approximate orthogonality between the source signals in this model, and then use blind source separation to exploit the orthogonality.
Simulation and experiment results both verify that our proposed algorithm achieves high angular resolution as distributed arrays without position errors, inversely proportional to the whole array aperture. 

\end{abstract}

\begin{IEEEkeywords}
Partly calibrated array, high angular resolution, orthogonality, blind source separation.
\end{IEEEkeywords}

\IEEEpeerreviewmaketitle

\section{Introduction}
\label{sec:introduction}

Direction finding using antenna arrays plays a fundamental role in various fields of signal processing such as radar, sonar, and astronomy \cite{526899}. 
High angular resolution is a key objective in these fields, and requires a large array aperture \cite{van2004optimum}. 
However, a large aperture usually means high system complexity, poor mobility, and high cost.

A promising alternative to a single large-aperture array is to partition the whole array into distributed subarrays, and fuse their measurements in a coherent way in order to achieve the same angular resolution as the whole array \cite{4102876,WANG2004131,jenn2009distributed,anton2017analysis,SUN2017122,mghabghab2020self}. 
A premise of coherently fusing the measurements from subarrays is the accurate position of each array element.
Positions of array elements within a subarray are easy to obtain, however there are generally inevitable position errors between subarrays, resulting in challenges to perform coherent signal processing and achieve high angular resolution \cite{492542}. 
For fixed ground-based platforms, these position errors could be calibrated offline.
For example, black hole imaging uses atomic clocks to precisely calibrate the position errors between subarrays, and synthesizes a large aperture close to the earth diameter \cite{collaboration2019first}.
In this case, the synthetic array, with both intra- and inter- subarray well calibrated, is referred to as a \emph{fully calibrated array}.
However, in mobile platforms such as unmanned aerial vehicles (UAV), the movement of the platform makes it hard to have accurate subarray positions in real time. 
Such a distributed array is often referred to as \emph{partly calibrated array}  \cite{492542}, with the intra-subarray elements well calibrated but not the inter-subarray elements. 
Estimating the \ac{doa} of radiating sources with a partly calibrated array has received wide attention recently \cite{COBRAS18,adler2019direct,8453887}, and is the focus of this paper. 

Early approaches for partly calibrated arrays first estimate the directions with each subarray separately, and then fuse the estimates to improve accuracy \cite{1164706,stoica1995decentralized,sarvotham2005distributed}. 
Since signals from different subarrays are not coherently processed, the angular resolution is limited to each single subarray instead of the synthesized aperture. 

More advanced algorithms jointly process received signals of all the subarrays. 
These techniques can be divided into two categories: correlation domain and direct data domain algorithms. Correlation domain techniques first calculate the covariance matrix of the received signals, separate the signal and noise subspaces from the covariance matrix, and use the subspaces to indicate the \acp{doa}. These subspace algorithms \cite{swindlehurst1992multiple,RARE02,exRARE04,5947005} are shown to exploit the whole array aperture and achieve high angular resolution.  
However, correlation domain approaches require a large number of snapshots to accurately estimate the covariance matrix and usually assume that the radiating sources are independent, which may not be satisfied in practice. 
Moreover, many snapshots increase the burden of data transmission and the processing delay, which has a significant impact on system performance. 
Particularly in mobile and time-varying scenarios \cite{Ding16}, the scenario changes rapidly during long-time observations, which leads to model mismatch and performance loss.
Hence, it is of significance to achieve high angular resolution with fewer snapshots, even a single snapshot.
However, few snapshots are not typically sufficient to estimate the covariance matrix correctly.

Direct data domain algorithms are often preferred compared to correlation domain algorithms. This is because direct data domain algorithms directly estimate the \acp{doa} by exploiting some prior information on the sources instead of their statistical properties learned from the received signals.  
Among these techniques, sparse recovery approaches have attracted interest in recent years \cite{6882328,COBRAS18}. 
These algorithms exploit the sparsity of sources and estimate \ac{doa}s by solving a block- and rank-sparse optimization problem.
They have tractable complexity and are shown to be less sensitive than correlation domain techniques to few snapshots and correlated sources.
However, they have poor performance in the single-snapshot case \cite{COBRAS18}.



To achieve high-resolution direction finding in partly calibrated arrays with only a single snapshot, we propose to transform the signal model into a multiple-measurement model \cite{eldar2009robust,eldar2010block}, where each measurement corresponds to the single-snapshot data of a subarray.
In this way, we encode the high-resolution capability into the phase relationship between the measurements.
We then show that there is approximate orthogonality between the measurements of different sources.
By exploiting the orthogonality, the phase offsets between subarrays are recovered to enhance angular resolution.

We use a \ac{bss} \cite{choi2005blind} algorithm called \ac{jade} \cite{cardoso1993blind} to exploit the orthogonality in our scenarios.
\ac{bss} is a typical tool that makes use of signal characteristics (independence) between measurements in multiple-measurement models.
Based on \ac{bss} ideas, our study illustrates that the cost function of \ac{jade} can well characterize the signal characteristics (orthogonality) between the measurements of sources in partly calibrated arrays. 
We thus apply \ac{jade} to exploit the orthogonality, and then use the output of \ac{jade} for phase offset recovery between subarrays, which leads to direction finding with high angular resolution.
Not only the simulation, but also experiment results verify that our proposed algorithm achieves high angular resolution, inversely proportional to the whole array aperture, with only a single snapshot.

The rest of this paper is organized as follows: Section~\ref{sec:signalmodel} introduces the signal model of partly calibrated arrays, and shows how to transform it into a multiple-measurement model. 
Section~\ref{sec:orthogonality} illustrates that there is approximate orthogonality between the source signals in partly calibrated arrays.
The proposed direction-finding algorithm exploiting the orthogonality is detailed in Section~\ref{sec:proposedmethod}. 
We discuss the relationship between the orthogonality and angular resolution in Section~\ref{sec:performanceyanalysis}.
Numerical simulation and experimental results are presented in Section~\ref{sec:simulation}, followed by a conclusion in Section~\ref{sec:conclusion}. 

Notation: We use $ \mathbb{R} $ and $ \mathbb{C} $ to denote the sets of real and complex numbers, respectively.
Uppercase boldface letters denote matrices (e.g. $ \bm{C} $) and lowercase boldface letters denote vectors (e.g. $ \bm{c} $). The $(m,n)$-th element of a matrix $ \bm{C} $ is denoted by $ [\bm{C}]_{m,n} $, and the $n$-th column is represented by $ [\bm{C}]_{n} $.  
We use $ {\rm trace}(\cdot) $ to indicate the trace of a matrix and $ {\rm diag}(\bm{a}) $ to represent a matrix with diagonal elements given by $ \bm{a} $.
The  conjugate, transpose, and  conjugate transpose operators are denoted by $ ^*,^T,^H $, respectively. 
The amplitude of a scalar, the $ l_2 $ norm of a vector and the Frobenius norm of a matrix are represented by $ |\cdot| $, $ \Vert \cdot \Vert_2 $ and $ \Vert \cdot \Vert $, respectively.
We use $ \rm{vec}(\cdot) $ for matrix vectorization.
The Hadamard product is written as $ \odot $, and $ \equiv $ is the definition symbol. 
We denote the imaginary unit for complex numbers by $ \jmath=\sqrt{-1} $.

\section{Signal model}
\label{sec:signalmodel}

In this section, we first introduce the signals of the partly calibrated model in Subsection \ref{subsec:pc}. 
For comparison, we introduce the signals of fully calibrated arrays and discuss angular resolution in Subsection \ref{subsec:com}.
Next, we transform the single-snapshot signals of partly calibrated arrays into a multiple-measurement model in Subsection \ref{subsec:trans}.
Finally, we explain the challenges and motivation of recovering directions with high angular resolution from this multiple-measurement model in Subsection \ref{subsec:motivation}. 

\subsection{Partly calibrated model}
\label{subsec:pc}

Consider $ K $ isotropic linear antenna subarrays, each composed of $ M_k $ sensors for $ k=1,\dots,K $, where $ M\equiv\sum_{k=1}^K M_k $.
We construct a planar Cartesian coordinate system, where the array is on the x-axis.
For these subarrays, the partly calibrated model assumes the precisely known intra-subarray displacement and unknown inter-subarray displacement. 
Particularly, denote by $ \xi_k\in\mathbb{R} $ the unknown inter-subarray displacement of the first (reference) sensor in the $ k $-th subarray relative to the the first sensor in the 1-st subarray for $ k=1,\dots,K $, thus, $ \xi_1=0 $. 
We use $ \bm{\xi}=[\xi_2,\dots,\xi_K]^T\in\mathbb{R}^{(K-1)\times1} $ to represent the inter-subarray displacement vector.
Denote by $ \eta_{k,m}\in\mathbb{R} $ the known intra-subarray displacement of the $ m $-th sensor relative to the 1-st sensor in the $ k $-th subarray for $ m=1,\dots,M_k, k=1,\dots,K $, thus, $ \eta_{k,1}=0 $. 
Assume that the $ M $ sensors share a common sampling clock, or the clocks have been synchronized exactly.
We illustrate the partly calibrated model in  Fig.~\ref{fig:geometry}.

\begin{figure}[!htbp]
\centering
\includegraphics[width=3in]{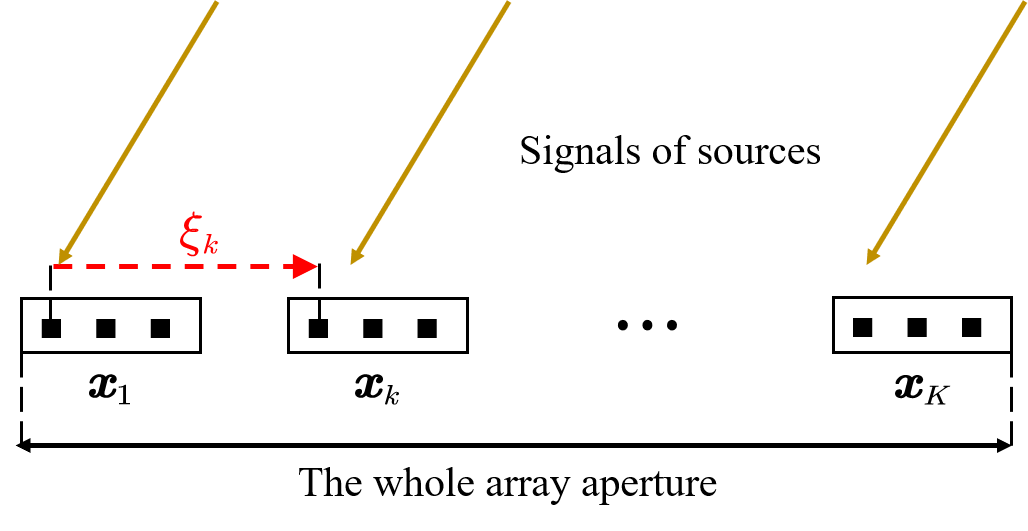}
\caption{The partly calibrated model, where $ \bm{\xi} $ is unknown.}
\label{fig:geometry} 
\end{figure}

There are $ L $ far-field \cite{8125734}, closely spaced emitters impinging narrow-band signals onto the whole array from different directions $ \theta_1,\dots,\theta_L $.
Denote the direction vector by $ \bm{\theta}=[\theta_1,\dots,\theta_L]^T\in\mathbb{R}^{L\times1} $.
In general cases, we assume $ \bm{\theta}\geqslant\bm{0} $ or $ \bm{\theta}\leqslant\bm{0} $.
The steering vector of the $ k $-th subarray corresponding to an emitter at direction $ \theta $ is given by
\begin{equation}
\label{equ:steervec}
\bm{a}_k(\theta) = \bm{b}_k(\theta)\phi_k(\theta,\xi)\in\mathbb{C}^{M_k\times1},
\end{equation}
where $ \phi_k(\theta,\xi)={\rm exp}\left(\jmath\frac{2\pi}{\lambda}\xi_k\sin\theta\right) $ is an unknown phase offset between the $ k $-th and the 1-st subarray, $ \lambda $ denotes the wavelength of the transmitted signals by emitters and the vector $ \bm{b}_k(\theta)\in\mathbb{C}^{M_k\times1} $ is defined by 
\begin{align}
\label{equ:vecb}
[\bm{b}_k(\theta)]_m = {\rm exp}\left(\jmath\frac{2\pi}{\lambda}\eta_{k,m}\sin\theta\right), m=1,\dots,M_k.
\end{align}
Here $ \bm{b}_k(\theta) $ is a known function of $ \theta $ in contrast to the unknown phase offset $ \phi_k(\theta,\xi) $ for $ k=1,\dots,K $.

In this paper, we assume that only a single snapshot is available and all the subarrays sample at the same time. Therefore, the signal received by the $ k $-th subarray is a summation of signals transmitted by the $ L $ emitters, given by
\begin{equation}
\label{equ:receivedsignals}
\bm{x}_k = \sum_{l=1}^L s_l\bm{a}_k(\theta_l)+\bm{n}_k\in\mathbb{C}^{M_k\times 1},
\end{equation}
where $ s_l $ represents an unknown complex coefficient and $ \bm{n}_k $ represents noise for $ k=1,\dots,K $.
By substituting \eqref{equ:steervec} into \eqref{equ:receivedsignals}, we rewrite the received signals \eqref{equ:receivedsignals} as
\begin{equation}
\label{equ:recast0}
\bm{x}_k = \sum_{l=1}^L s_l\bm{b}_k(\theta_l)\phi_k(\theta_l,\xi)+\bm{n}_k.
\end{equation}

\subsection{Fully calibrated model}
\label{subsec:com}

The fully calibrated arrays assume that $ \bm{\xi} $ in \eqref{equ:recast0} is exactly known or well calibrated.
By substituting \eqref{equ:vecb} and $ \phi_k(\theta,\xi)={\rm exp}\left(\jmath\frac{2\pi}{\lambda}\xi_k\sin\theta\right) $ to \eqref{equ:recast0}, we have
\begin{align}
\label{equ:fully1}
\bm{x}_k=\bm{C}_k(\bm{\theta})\bm{s}+\bm{n}_k,
\end{align}
where $ [\bm{C}_k(\bm{\theta})]_{m,l}={\rm exp}\left(\jmath\frac{2\pi}{\lambda}(\eta_{k,m}+\xi_k)\sin\theta_l\right) $ for $ m=1,\dots,M_k $ and $ \bm{s}=[s_1,\dots,s_L]^T\in\mathbb{C}^{L\times1} $.
Denote $ \tilde{\bm{x}}_c=[\bm{x}_1^T,\dots,\bm{x}_K^T]^T\in\mathbb{C}^{M\times1} $, $ \bm{C}(\bm{\theta})=[\bm{C}_1^T(\bm{\theta}),\dots,\bm{C}_K^T(\bm{\theta})]^T\in\mathbb{C}^{M\times L}$ and $ \tilde{\bm{n}}_c=[\bm{n}_1^T,\dots,\bm{n}_K^T]^T\in\mathbb{C}^{M\times1} $. 
We then stack the $ K $ signals in \eqref{equ:fully1} together as 
\begin{align}
\label{equ:fully2}
\tilde{\bm{x}}_c=\bm{C}(\bm{\theta})\bm{s}+\tilde{\bm{n}}_c,
\end{align}
which is a typical single-snapshot signal model.
In \eqref{equ:fully2}, $ \xi_k $ and $ \eta_{k,m} $ are known and the unknowns in $ \bm{C}(\bm{\theta}) $ are only $ \bm{\theta} $.
Direction finding is to estimate $ \bm{\theta} $ given $ \tilde{\bm{x}}_c $, which can be achieved by existing classical algorithms such as \ac{music} \cite{LIAO201633} and \ac{cs} \cite{eldar2012compressed}.

Direction finding by fully calibrated arrays can achieve high angular resolution, inversely proportional to the whole array aperture.
This is because the fully calibrated model assumes completely known $ \bm{\xi} $ and the received signals are recast as \eqref{equ:fully2}. 
In this case, the subarrays can be regarded as sparsely distributed array elements in a single array with a large aperture, where the array manifold is denoted by $ \bm{C}(\bm{\theta}) $ in \eqref{equ:fully2}. 
In this single array, the positions of array elements are in the range from $ 0 $ to $ \eta_{K,M_K}+\xi_K $, yielding the whole array aperture.
Therefore, the angular resolution of \eqref{equ:fully2} is inversely proportional to the whole array aperture \cite{van2004optimum}.

\subsection{Transform into a multiple-measurement model}
\label{subsec:trans}


However, for the partly calibrated model, an accurate array manifold $ \bm{C}(\bm{\theta}) $ is not available due to the unknown $ \bm{\xi} $, which makes it hard to synthesize the subarrays into a single large-aperture array. 
The unknown $ \bm{\xi} $ introduces unknown phase errors $ \phi_k(\theta_l,\xi) $, which are related to both subarrays and sources, to the signal model.
Particularly, define $ \bm{\Phi}_k(\bm{\theta},\bm{\xi})={\rm diag}(\phi_k(\theta_1,\xi),\dots,\phi_k(\theta_L,\xi))\in\mathbb{C}^{L\times L} $ and $ \bm{B}_k(\bm{\theta})=[\bm{b}_k(\theta_1),\dots,\bm{b}_k(\theta_L)]\in\mathbb{C}^{M_k\times L} $ for $ k=1,\dots,K $. 
We then rewrite \eqref{equ:recast0} as
\begin{equation}
\label{equ:recast}
\bm{x}_k=\bm{B}_k(\bm{\theta})\bm{\Phi}_k(\bm{\theta},\bm{\xi})\bm{s}+\bm{n}_k,
\end{equation}
where $ \bm{\Phi}_k(\bm{\theta},\bm{\xi}) $ denotes the unknown phase error w.r.t. $ \bm{\xi} $. Compared with \eqref{equ:fully1}, the unknown $ \bm{\Phi}_k(\bm{\theta},\bm{\xi}) $ in \eqref{equ:recast} makes high-resolution direction finding a challenging problem in single-snapshot cases \cite{COBRAS18}.
For brevity, we denote $ \bm{\Phi}_k(\bm{\theta},\bm{\xi}) $ and $ \bm{B}_k(\bm{\theta}) $ by $ \bm{\Phi}_k $ and $ \bm{B}_k $, respectively.

There is a key question in partly calibrated arrays:

\emph{Can partly calibrated subarrays achieve angular resolution comparable to that of fully calibrated subarrays?}

\noindent
The answer is positive. 
The feasibility lies in the fact that the unknown phase offsets between subarrays, $ \phi_k(\theta_l,\xi) $, can be well recovered by self-calibration on $ \bm{\xi} $.
In this paper, we provide an algorithm to achieve high angular resolution performance.
We leave the theoretical analysis for future work.

To illustrate our approach, we consider a typical case in partly calibrated arrays, where the intra-subarray displacement of each subarray is the same. In this case, we transform the partly calibrated model into a multiple-measurement model \cite{eldar2009robust,eldar2010block}. Particularly, $ \{\eta_{k,\cdot} \}$ and $ \{M_k\} $ reduce to the same values for different $ k $, and are denoted by $  \bar{\eta}_{\cdot}=\eta_{k,\cdot} $ and $ \bar{M}=M_k $, respectively. We use $ \bm{\eta}=[\bar{\eta}_{2},\dots,\bar{\eta}_{\bar{M}}]^T\in\mathbb{R}^{(\bar{M}-1)\times1} $ to represent the intra-subarray displacement vector. 
The steering matrices $ \bm{B}_k $ become the same for all the subarrays, thus we denote $ \bar{\bm{B}} = \bm{B}_k\in\mathbb{C}^{\bar{M}\times L} $ for $ k=1,\dots,K $.
We assume that for each subarray, $ \bar{M}>L $, such that emitters are identifiable \cite{SPICE}. Under these assumptions, $ \bar{\bm{B}}$ is a full column rank matrix. Then, the received signals \eqref{equ:recast} are reorganized as 
\begin{align}
\label{equ:reorganized}
[\bm{x}_1,\dots,\bm{x}_K]&=\bar{\bm{B}}\left[\bm{\Phi}_1\bm{s},\dots,\bm{\Phi}_K\bm{s}\right]+[\bm{n}_1,\dots,\bm{n}_K] \nonumber \\
&=\bar{\bm{B}}\begin{bmatrix}
\bar{\bm{s}}_1^T \\
\vdots \\
\bar{\bm{s}}_L^T
\end{bmatrix}+[\bm{n}_1,\dots,\bm{n}_K],
\end{align}
where $ \bar{\bm{s}}_l=s_l\left[\phi_1(\theta_l,\xi),\dots,\phi_K(\theta_l,\xi)\right]^T\in\mathbb{C}^{K\times1} $ are viewed as the received signals from the $ l $-th emitter,  sampled $ K $ times by different subarrays. 
Denote $ \bm{X}=[\bm{x}_1,\dots,\bm{x}_K]\in\mathbb{C}^{\bar{M}\times K} $, $ \bm{S}=\left[\bar{\bm{s}}_1,\dots,\bar{\bm{s}}_L\right]^T\in\mathbb{C}^{L\times K} $ and $ \bm{N}=[\bm{n}_1,\dots,\bm{n}_K]\in\mathbb{C}^{\bar{M}\times K} $. Then \eqref{equ:reorganized} is recast as 
\begin{equation}
\label{equ:XBSN}
\bm{X}=\bar{\bm{B}}\bm{S}+\bm{N},
\end{equation}
where $ \bar{\bm{B}} $ is related to $ \{\bm{\theta},\bm{\eta}\} $, and $ \bm{S} $ is related to $ \{\bm{\theta},\bm{\xi},\bm{s}\} $.
In \eqref{equ:XBSN}, $ \{\bm{X},\bm{\eta} \}$ are known, and $ \{\bm{\theta},\bm{\xi},\bm{s},\bm{N} \}$ are unknown. In addition, $ L $ is assumed to be known by existing algorithms \cite{akaike1974new,schwarz1978estimating,wax1985detection}. 
The recast signal model \eqref{equ:XBSN} is a multiple-measurement model, where each row of $ \bm{S} $ represents the multiple measurements of the $ K $ subarrays. 
Direction finding is to estimate $ \bm{\theta} $ given $ \bm{X} $.

Note that $ \bm{S} $ in \eqref{equ:XBSN} is constructed by multiple measurements of subarrays instead of multiple time samples as typical multiple-measurement models, and has special signal structure.
Particularly, the entries of the $ l $-th row of $ \bm{S} $ have the same modulus $ |s_l| $ and different phase offsets $ \phi_k(\theta_l,\xi) $.
In \eqref{equ:XBSN}, high-resolution capability of partly calibrated arrays is encoded in this special signal structure of $ \bm{S} $, mainly in $ \phi_k(\theta_l,\xi) $.
The key issue of enhancing angular resolution lies in how to exactly recover $ \phi_k(\theta_l,\xi) $ by exploiting the signal structure of $ \bm{S} $.



\subsection{Challenges and motivation}
\label{subsec:motivation}

Existing algorithms encounter difficulties in achieving high angular resolution based on \eqref{equ:XBSN}.
This is because $ \bm{S} $ is a complex function of $ \{\bm{\theta},\bm{\xi},\bm{s}\} $ and exploiting the signal structure of $ \bm{S} $ to recover $ \phi_k(\theta_l,\xi) $ is not a trivial problem.
Particularly, \ac{music} algorithms \cite{schmidt1986multiple,LIAO201633} suffer from imprecise subspace structure of signals due to the unknown $ \bm{\xi} $.
Typical sparse recovery algorithms \cite{eldar2015sampling} divide the parameter space of $ \bm{\theta} $ into finite grids, construct a over-complete, grid-based dictionary of $ \bar{\bm{B}} $, and solve a convex lasso \cite{tibshirani1996regression} problem.
However, these algorithms ignore the phase relationship between subarrays $ \phi_k(\theta_l,\xi) $ and assume $ \bm{S} $ to be completely unknown, which corresponds to non-coherent processing which cannot achieve high angular resolution.
Therefore, we do not use gridding on \eqref{equ:XBSN}.
To enhance angular resolution, more efficient constraints on $ \bm{S} $ are required in the problem formulation.
However, it is hard to propose a proper constraint that fully characterizes the signal structure, and the introduction of such constraints increase the difficulty of the solution.

Our goal is to find a way that both makes full use of the signal structure and is tractable to solve.
As mentioned above, the key of high angular resolution lies in the exact recovery of $ \phi_k(\theta_l,\xi) $, which are the phases of $ \bm{S} $ in \eqref{equ:XBSN}.
Therefore, we propose to exploit a specific characteristic of $ \bm{S} $ to recover $ \phi_k(\theta_l,\xi) $ first, and then find the directions based on the phase estimates. 
The characteristic we consider is approximate orthogonality between the source signals of partly calibrated arrays, which is detailed next.

\section{Approximate orthogonality between the rows of \texorpdfstring{$\bm{S}$}{e}}
\label{sec:orthogonality}

In this section, we first analyze the statistics of the sample covariance of $ \bm{S} $, $ \hat{\bm{R}}_{\bar{\bm{s}}}\equiv\frac{1}{K}\bm{S}\bm{S}^H $, in Subsection \ref{subsec:ortho1}, where $ \bm{S} $ denotes the source signals in \eqref{equ:XBSN}. 
Then, we show that there is approximate orthogonality between the rows of $ \bm{S} $ in partly calibrated arrays in Subsection \ref{subsec:ortho2}.

\subsection{Statistics of \texorpdfstring{$\hat{\bm{R}}_{\bar{\bm{s}}}\equiv\frac{1}{K}\bm{S}\bm{S}^H$}{e}}
\label{subsec:ortho1}

First, we calculate the elements of $ \hat{\bm{R}}_{\bar{\bm{s}}}\equiv\frac{1}{K}\bm{S}\bm{S}^H $.
By substituting the definitions of $ \bm{S} $, $ \bar{\bm{s}}_l $ and  $ \phi_k(\theta,\xi) $ into $ \hat{\bm{R}}_{\bar{\bm{s}}}$, the $(i,j)$-th element of $ \hat{\bm{R}}_{\bar{\bm{s}}} $, denoted by $R_{i,j}$ for abbreviation, is
\begin{align}
\label{equ:covelements}
R_{i,j} \equiv [\hat{\bm{R}}_{\bar{\bm{s}}}]_{i,j}&=\frac{s_is_{j}^*}{K} \sum_{k=1}^K \phi_k(\theta_i,\xi)\phi_k^*(\theta_{j},\xi), \nonumber \\
&=\frac{s_is_{j}^*}{K} \sum_{k=1}^K e^{\jmath\frac{2\pi}{\lambda}\xi_k\left(\sin\theta_i-\sin\theta_{j}\right)},
\end{align}
where $ R_{i,j}=R_{j,i}^* $ for $ i,j=1,\dots,L $.
The diagonal elements of $ \hat{\bm{R}}_{\bar{\bm{s}}} $ are $ R_{i,i}=|s_i|^2 $.
We assume that the source intensities are all equal, $|s_i| = 1$,  $ i=1,\dots,L $ in the sequel.
Under this assumption, the intensities of the off-diagonal elements rely on the inter-subarray displacements $ \bm{\xi} $ and the directions $ \bm{\theta} $.

Denote the whole array aperture and the intra-subarray element spacing by $ D $ and $ d $, respectively. 
To analyze the statistics of $ \hat{\bm{R}}_{\bar{\bm{s}}} $, we consider a common practical scenario that the subarrays are randomly, uniformly placed in $ [0,D] $, i.e., $ \xi_k \sim \mathcal{U}\left[0,D\right] $ for $ k=1,\dots,K $.
The corresponding geometry is shown in Fig.~\ref{fig:randomgeometry}.
\begin{figure}[!htbp]
\centering
\includegraphics[width=3in]{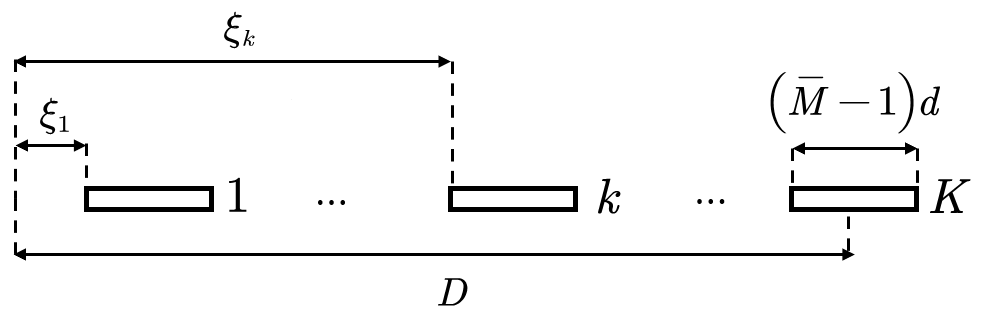}
\caption{The geometry of the randomly, uniformly distributed subarrays.}
\label{fig:randomgeometry} 
\end{figure}

This typical case leads to the following proposition, a similar analysis can be found in \cite{6657792}. 

\begin{prop}
\label{prop:prop2}
When $ \xi_k \sim \mathcal{U}\left[0,D\right] $ for $k=1,\dots,K $, we have
\begin{equation}
\label{equ:ERs}
\left|\mathbb{E}\left[R_{i,j}\right]\right|=\left|\frac{\sin\rho_{i,j}}{\rho_{i,j}}\right|,
\end{equation}
\begin{equation}
\label{equ:propcovcal}
\mathbb{E}\left[\left|R_{i,j}\right|^2\right]=\frac{1}{K}+\left(1-\frac{1}{K}\right)\left|\frac{\sin \rho_{i,j}}{\rho_{i,j}}\right|^2,
\end{equation}
where $ \rho_{i,j}=\frac{\pi D}{\lambda}(\sin\theta_i-\sin\theta_{j}) $, $ i,j=1,\dots,L $ and $ i\ne j $.
\end{prop}
\begin{proof}
See Appendix~\ref{app:prop2}.
\end{proof}

The term $ \left|\frac{\sin\rho_{i,j}}{\rho_{i,j}}\right| $ plays an important role in \eqref{equ:ERs} and \eqref{equ:propcovcal}. 
We draw the curve of $ \left|\frac{\sin\rho_{i,j}}{\rho_{i,j}}\right| $ w.r.t. $ (\sin\theta_i-\sin\theta_j)/\Delta $ in Fig.~\ref{fig:prop2}, where $ \Delta\equiv\lambda/D $ denotes the empirical angular resolution of the whole array.

\begin{figure}[!htbp]
\centering
\includegraphics[width=3in]{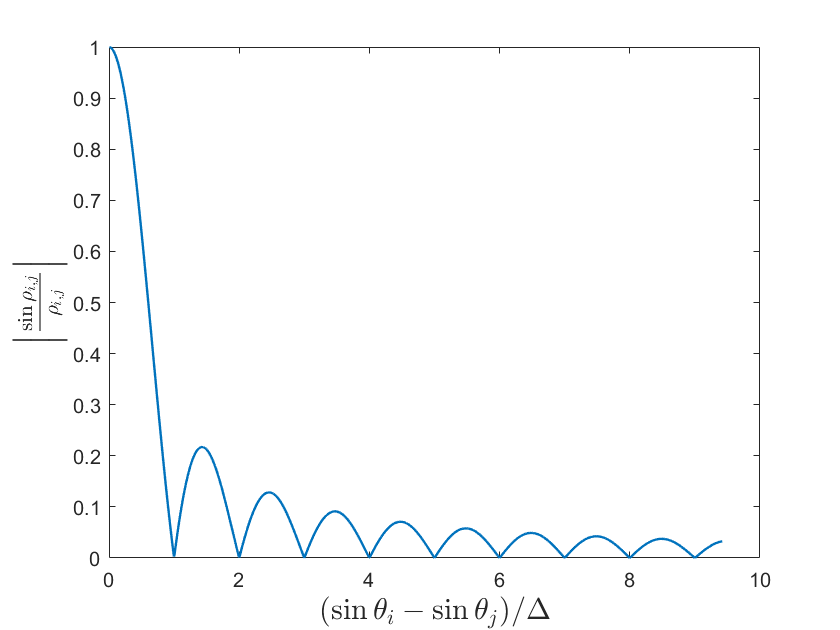}
\caption{The off-diagonal elements $ |R_{i,j}| $ in a statistical sense.}
\label{fig:prop2} 
\end{figure}

From Fig.~\ref{fig:prop2}, we see that the curve is composed of a main lobe and several side lobes.
When the interval $\sin \theta_i - \sin \theta_j$ equals $\Delta$, the intensity $\left|R_{i,j}\right|$ reaches its first zero point, which is related to the angular resolution $ \Delta $. 
When the interval is in the region $[0,\Delta)$, the distributed array cannot distinguish sources at such close interval and $\left|R_{i,j}\right|$ has large values.  
When the interval belongs to the region $ [\Delta, \infty) $, $\left|R_{i,j}\right|$ is typically small.
In this region, the zero points of $\left|R_{i,j}\right|$  occur in a period of $\Delta$ and there are side lobes between these zero points. 


\subsection{Approximate orthogonality in \texorpdfstring{$\hat{\bm{R}}_{\bar{\bm{s}}}$}{e}}
\label{subsec:ortho2}

Based on the statistics of $ \hat{\bm{R}}_{\bar{\bm{s}}} $, there is approximate orthogonality in partly calibrated arrays. Here, orthogonality is defined as the orthogonality between the rows of $ \bm{S} $ in \eqref{equ:XBSN}. 
Particularly, we say there is orthogonality in partly calibrated arrays if the off-diagonal elements of $ \hat{\bm{R}}_{\bar{\bm{s}}} $ are 0, i.e.,
\begin{align}
\label{equ:rowortho}
\hat{\bm{R}}_{\bar{\bm{s}}}=
\begin{bmatrix}
|s_1|^2 & & 0 \\
& \ddots & \\
0 & & |s_L|^2
\end{bmatrix}.
\end{align}
In Proposition \ref{prop:prop2} and Fig.~\ref{fig:prop2}, we show that when the source separation exceeds one resolution cell, the off-diagonal elements of $ \hat{\bm{R}}_{\bar{\bm{s}}} $ are small relative to the diagonal elements in a statistical sense.
In other words, there is approximate orthogonality between the rows of $ \bm{S} $ in \eqref{equ:XBSN}, and
\begin{align}
\label{equ:rowortho_approximate}
\hat{\bm{R}}_{\bar{\bm{s}}}\approx
\begin{bmatrix}
|s_1|^2 & & 0 \\
& \ddots & \\
0 & & |s_L|^2
\end{bmatrix}.
\end{align}

We may also understand the orthogonality in partly calibrated arrays using the concept of coherence \cite{tropp2017mathematical}, given by
\begin{align}
\label{equ:coherence}
\mu\equiv\underset{i\ne j}{\rm max} \frac{| \bar{\bm{s}}_j^H\bar{\bm{s}}_i |}{\Vert \bar{\bm{s}}_i \Vert\Vert \bar{\bm{s}}_j \Vert},
\end{align}
where $ 0\leqslant\mu\leqslant1 $.
In partly calibrated arrays, the coherence $ \mu $ corresponds to the largest modulus of the off-diagonal elements of $ \hat{\bm{R}}_{\bar{\bm{s}}} $.
Orthogonality means that $ \mu=0 $ and approximate orthogonality means that $ 0<\mu\ll1 $.

In summary, when the source separation is within a resolution cell, there is no obvious orthogonality between the source signals. 
When the source separation exceeds a resolution cell, the modulus of the off-diagonal $ R_{i,j} $ is small in a statistical sense, yielding approximate orthogonality. This phenomenon not only illustrates the orthogonality in partly calibrated arrays, but also builds a connection between orthogonality and angular resolution, which will be discussed in Section \ref{sec:performanceyanalysis}.

\section{Joint Approximate Diagonalization}
\label{sec:proposedmethod}

In this section, we introduce our algorithm. 
First, we propose to use \ac{bss} ideas to exploit the orthogonality, and recover the phase offsets between subarrays in Subsection \ref{subsec:phaseoffset}.
Then, based on the estimated phase offsets, beamforming algorithms for direction finding are discussed in Subsection \ref{subsec:direction}.

\subsection{Phase offset recovery}
\label{subsec:phaseoffset}

In \eqref{equ:XBSN}, when $ \bar{\bm{B}} $ is unknown, how to exploit the orthogonality between the rows of $ \bm{S} $ is a challenging problem.
Here, we propose a solution to this problem.
Particularly, we introduce \ac{bss} techniques to the self-calibration problems in distributed arrays.
To this end, we first review basic \ac{bss} ideas. 
Then, we show how a specific \ac{bss} algorithm exploits the orthogonality of $ \bm{S} $ in \eqref{equ:XBSN}. 
Finally, we show how we recover the phase offsets between subarrays. 

\subsubsection{Review of BSS} 
\ac{bss} problems usually consider the following model:
\begin{equation}
\label{equ:YCHN}
\bm{Y}=\bm{C}\bm{H}+\bm{N},
\end{equation}
where $ \bm{Y}\in \mathbb{C}^{N \times T_s} $ contains the received signals, $ \bm{C}\in\mathbb{C}^{N\times L} $ is an unknown, full-rank measurement matrix, $ \bm{H}=[\bm{h}_1,\dots,\bm{h}_L]^T\in \mathbb{C}^{L \times T_s}$ are the unknown original source signals, $ \bm{N}\in \mathbb{C}^{N \times T_s}$ are the unknown white noises, $ T_s $, $ N $ and $ L $ denote the number of samples, measurements and sources, respectively, and $ L<N $. 
The purpose of \ac{bss} is to recover $ \bm{H} $ from $ \bm{Y} $ with unknown $ \bm{C} $. 

In order to make the recovery of $\bm H$ possible, \ac{bss} usually imposes assumptions on the independence between the original source signals $ \bm{h}_l, l=1,\dots,L $.
Under these assumptions, \ac{bss} formulates an optimization problem exploiting the independence of $\bm{h}_l$ to find the linear `inverse' matrix of $ \bm{C} $ in \eqref{equ:YCHN} as $ \widehat{\bm{C}}^{-1} $, where $ \widehat{\bm{C}}^{-1}\bm{Y} $ yields the recovery of $ \bm{H} $.
Note that there are unavoidable ambiguity between $\bm H$ and $\bm C$ in \ac{bss} problems. For example, for any nonzero diagonal matrix $\bm{Q}\in\mathbb{C}^{L\times L} $, we have
$\bm{C}\bm{H}=\left(\bm{C}\bm{Q}\right)\left(\bm{Q}^{-1}\bm{H}\right)$. 
Therefore, without loss of generality, it is typically assumed that the source signal has unit power, $ \mathbb{E}[h_l(t)h_l^*(t)] =1$, where $ h_l(t) $ denotes the continuous signals.
More details on \ac{bss} can be found in \cite{choi2005blind}.

\subsubsection{How BSS exploits the orthogonality}

Here we introduce a specific \ac{bss} algorithm, called \ac{jade} \cite{cardoso1993blind}, and show how \ac{jade} exploits the orthogonality in partly calibrated arrays from its cost function.
\ac{jade} assumes unknown measurement matrix $ \bm{C} $ in \eqref{equ:YCHN}, which is compatible with the unknown measurement matrix $ \bar{\bm{B}} $ in \eqref{equ:XBSN}.

The cost function of \ac{jade} is given by \cite{cardoso1993blind}
\begin{equation}
\label{equ:jade_fs}
    f_s(\bm{S})=\sum_{r,p,q=1,\dots,L} \Big|\widehat{\text{Cum}}\left([\bm{S}^T]_r,[\bm{S}^H]_r,[\bm{S}^T]_p,[\bm{S}^H]_q\right)\Big|^2,
\end{equation}
where the high-order cumulant $ \widehat{\text{Cum}}(\cdot) $ is denoted by
\begin{align}
\label{equ:jade_cum}
&\widehat{\text{Cum}}(\bm{\epsilon}_a,\bm{\epsilon}_b,\bm{\epsilon}_c,\bm{\epsilon}_d)= \frac{1}{K}(\bm{\epsilon}_a\odot\bm{\epsilon}_b)^T(\bm{\epsilon}_c\odot\bm{\epsilon}_d) \nonumber \\
&-\frac{1}{K^2}\left(\bm{\epsilon}_a^T\bm{\epsilon}_b\bm{\epsilon}_c^T\bm{\epsilon}_d+\bm{\epsilon}_a^T\bm{\epsilon}_c\bm{\epsilon}_b^T\bm{\epsilon}_d+\bm{\epsilon}_a^T\bm{\epsilon}_d\bm{\epsilon}_b^T\bm{\epsilon}_c\right),
\end{align}
and $ \bm{\epsilon} $ corresponds to $ [\bm{S}^T]_l=\bar{\bm{s}}_l $ in \eqref{equ:reorganized}.
Based on the definition of $ R_{i,j} $ in \eqref{equ:covelements}, we have \begin{align}
\label{equ:ortho_jade_Rij}
R_{i,j}\equiv\frac{1}{K}\bar{\bm{s}}_j^H\bar{\bm{s}}_i=\frac{1}{K}[\bm{S}^T]^H_j[\bm{S}^T]_i,
\end{align}
which corresponds to $ \bm{\epsilon}_j^T\bm{\epsilon}_i $ in \eqref{equ:jade_cum} for $ i,j=1,\dots,L $.

The expression of $ f_s(\bm{S}) $ in \eqref{equ:jade_fs} w.r.t high-order cumulants is complex, however, it can be simplified due to the special signal structure of partly calibrated arrays.
Particularly, since $ \bar{\bm{s}}_l $ in \eqref{equ:reorganized} has constant modulus $ |s_l| $, $ R_{l,l}=\frac{1}{K}\bar{\bm{s}}_l^H\bar{\bm{s}}_l=|s_l|^2 $ and $ \bar{\bm{s}}_l\odot\bar{\bm{s}}_l^H=|s_l|^2\cdot[1,1,\dots,1]^T $.
Based on the above properties, substituting \eqref{equ:ortho_jade_Rij} to \eqref{equ:jade_fs}, we have
\begin{align}
\label{equ:fs_jade_Rij}
f_s(\bm{S})&=\sum_{r,p,q}\left||s_r|^2R_{p,q}-|s_r|^2R_{p,q}-\tilde{R}_{r,p}\tilde{R}_{r,q}^*-R_{r,q}R_{p,r}\right|^2, \nonumber \\
&=\sum_{r,p,q}\left|\tilde{R}_{r,p}\tilde{R}_{r,q}^*+R_{r,q}R_{p,r}\right|^2,
\end{align}
where
\begin{align}
\label{equ:new_R}
\tilde{R}_{i,j}\equiv\frac{1}{K}\bar{\bm{s}}_j^T\bar{\bm{s}}_i=\frac{s_is_{j}}{K} \sum_{k=1}^K e^{\jmath\frac{2\pi}{\lambda}\xi_k\left(\sin\theta_i+\sin\theta_{j}\right)},
\end{align}
for $ i,j=1,\dots,L $.
Note that $ \tilde{R}_{i,j} $ is similar to $ R_{i,j} $, and the property of $ R_{i,j} $ in Proposition \ref{prop:prop2} when $ |\sin\theta_j-\sin\theta_i|\geqslant\Delta $ can be directly extended to $ \tilde{R}_{i,j} $ when $ |\sin\theta_j+\sin\theta_i|\geqslant\Delta $.

Based on \eqref{equ:fs_jade_Rij}, we now provide an intuitive explanation to how \ac{jade} exploits the orthogonality in partly calibrated arrays.
In Proposition \ref{prop:prop2} and Fig.~\ref{fig:prop2}, we illustrate that when the source separation exceeds one resolution cell, we have small $ R_{i,j} $ in a statistical sense, yielding approximate orthogonality.
We extend this property to $ \tilde{R}_{i,j} $ and have
\begin{align}
\label{equ:jade_and_ortho}
\left.
\begin{aligned}
|\sin\theta_i-\sin\theta_j|\geqslant\Delta, |R_{i,j}|\approx0 \\
|\sin\theta_i+\sin\theta_j|\geqslant\Delta, |\tilde{R}_{i,j}|\approx0
\end{aligned}
\right\}\Rightarrow f_s(\bm{S})\approx 0, 
\end{align}
which means that smaller $ |R_{i,j}| $ and $ |\tilde{R}_{i,j}| $ yield smaller $ f_s(\bm{S}) $ approximately. 
When $ \bm{\theta}\geqslant\bm{0} $ or $ \bm{\theta}\leqslant\bm{0} $ as assumed, 
$ |\sin\theta_i-\sin\theta_j|\geqslant\Delta $ implies $ |\sin\theta_j+\sin\theta_i|\geqslant\Delta $.
Therefore, when the sources are resolvable ($ |\sin\theta_i-\sin\theta_j|\geqslant\Delta $), minimizing $ f_s(\bm{S}) $ is likely to have an optimal result of $ \bm{S} $
with $ |R_{i,j}|\approx 0 $.
This illustrates that the cost function of \ac{jade} is feasible to exploit the orthogonality in partly calibrated arrays.
The optimal solution of \ac{jade} is not guaranteed to be global from \eqref{equ:fs_jade_Rij}. 
However, \ac{jade} is verified to achieve good performance by simulations in Section \ref{sec:simulation}.
This shows the potential of using \ac{bss} techniques to enhance angular resolution in partly calibrated arrays.


In most practical scenarios, the orthogonality is approximately satisfied, and the performance of \ac{jade} is affected by the level of the approximation.
Intuitively speaking, higher level of orthogonality corresponds to a lower cost function of \ac{jade}, which likely yields better performance of \ac{jade} in phase recovery.
This conclusion is verified by simulations in Fig.~\ref{fig:func_angle} and Fig.~\ref{fig:RMSE_angle} in Section \ref{sec:simulation}.
The specific algorithm of \ac{jade} is detailed in Appendix \ref{app:JADEalg} for reference.

\subsubsection{Phase offset estimation}

We showed how \ac{jade} exploits the orthogonality.
Next we discuss how to use \ac{jade} for phase offset estimation.

To this end, we input the received signals of partly calibrated arrays in \eqref{equ:XBSN} to \ac{jade}, which outputs an estimate of $ \bm{S} $, denoted by $ \hat{\bm{S}} $. 
Due to the unavoidable ambiguity in the recovery of \ac{bss}, it is assumed that the source signal has unit power in \ac{jade}.
However, the estimate $ \hat{\bm{S}} $ usually has fluctuating modulus. 
Therefore, we normalize $\hat{\bm{S}}$ and estimate the phase offsets as
\begin{equation}
\label{equ:phaseest}
\hat{\phi}_{l,k}=\frac{[\hat{\bm{S}}]_{l,k}}{\left|[\hat{\bm{S}}]_{l,k}\right|},
\end{equation}
for $ l=1,\dots,L $ and $ k=1,\dots,K $. 
In the sequel, we denote $ \hat{\bm{\Phi}}_k={\rm diag}\left(\hat{\phi}_{1,k},\dots,\hat{\phi}_{l,k}\right)\in\mathbb{C}^{L\times L} $.

\subsection{Direction finding}
\label{subsec:direction}

When the phase offsets between subarrays is recovered by \eqref{equ:phaseest}, direction finding reduces to a coherent array signal processing problem: Recover $ \bm{\theta} $ from the following signals,
\begin{equation}
\label{equ:coheretsp}
\bm{X}=\bar{\bm{B}}(\bm{\theta})\left[\hat{\bm{\Phi}}_1\bm{s},\dots,\hat{\bm{\Phi}}_K\bm{s}\right]+\bm{N}.
\end{equation}
This problem can be solved by many classical algorithms.
In this subsection, we introduce two representative techniques for direction finding.

\subsubsection{Matched filtering}

We begin with \ac{mf} for direction finding.
Consider the $ L=1 $ case, where \eqref{equ:coheretsp} is simplified as 
\begin{equation}
\label{equ:L=1simply}
\bm{X}=\bar{\bm{b}}(\theta_1)\left[\hat{\phi}_{1,1},\dots,\hat{\phi}_{1,K}\right]s+\bm{N},
\end{equation}
and the steering vector $ \bar{\bm{b}}(\theta)\in\mathbb{C}^{\bar{M}\times1} $ is denoted by
$ \bar{\bm{b}}(\theta)=\left[1,{\rm exp}\left(\jmath\frac{2\pi}{\lambda}\eta_{2}\sin\theta\right),\dots,{\rm exp}\left(\jmath\frac{2\pi}{\lambda}\eta_{\bar{M}}\sin\theta\right)\right]^T $.
For \eqref{equ:L=1simply}, the \ac{mf} algorithm is to find $ \theta_1 $ by maximizing the following cost function,
\begin{align}
\label{equ:L=1MF}
\hat{\theta}_1&=\underset{\theta\in(-\frac{\pi}{2},\frac{\pi}{2})}{\rm argmax}\ \left| {\rm vec}^H(\bm{X}){\rm vec}\left(\bar{\bm{b}}(\theta)\left[\hat{\phi}_{1,1},\dots,\hat{\phi}_{1,K}\right]\right) \right| \nonumber \\
&=\underset{\theta\in(-\frac{\pi}{2},\frac{\pi}{2})}{\rm argmax}\ \left| \sum_{k=1}^K \bm{x}_k^H\left(\bar{\bm{b}}(\theta)\hat{\phi}_{1,k}\right) \right|.
\end{align}
The \ac{mf} algorithm is directly applicable to the case when $ L\geqslant2 $.
Particularly, directions $ \bm{\theta} $ are estimated as
\begin{equation}
\label{equ:MF}
\hat{\theta}_l=\underset{\theta\in(-\frac{\pi}{2},\frac{\pi}{2})}{\rm argmax}\ \left| \sum_{k=1}^K \bm{x}_k^H\left(\bar{\bm{b}}(\theta)\hat{\phi}_{l,k}\right) \right|,
\end{equation}
for $ l=1,\dots,L $. 
The optimization problem \eqref{equ:MF} can be directly solved by conducting $ L $ individual one-dimensional searches on the direction range $ (-\pi/2,\pi/2) $.

The \ac{mf} algorithm \eqref{equ:MF} is easy to implement. 
However, in the multi-source ($ L\geqslant2 $) case, the estimation accuracy of $ \hat{\theta}_l $ by \eqref{equ:MF} is affected by the other $ L-1 $ source signals in $ \bm{x}_k $, i.e., $ \sum_{r\ne l}^L s_r\bar{\bm{b}}(\theta_r)\phi_k(\theta_r,\xi) $ in \eqref{equ:reorganized} for $ l=1,\dots,L $.
Therefore, for better performance, we introduce an alternative algorithm to solve this problem.

\subsubsection{Nonlinear least squares}

Here, we introduce the \ac{nls} algorithm \cite{gill1978algorithms} that jointly estimates $ \{\bm{\theta},\bm{s}\} $. 
Particularly, we consider the following \ac{nls} optimization problem w.r.t. $ \{\bm{\theta},\bm{s}\} $,
\begin{equation}
\label{equ:directionfinding}
\underset{\bm{\theta},\bm{s}}{\rm min}\  C(\bm{\theta},\bm{s})=\sum_{k=1}^K \left\Vert  \bm{x}_k-\bar{\bm{B}}(\bm{\theta})\hat{\bm{\Phi}}_k\bm{s} \right\Vert_2^2.
\end{equation}
To solve \eqref{equ:directionfinding}, we use gradient descend algorithms and alternatively carry out the following two steps until convergence:
\begin{align}
\label{equ:LSofs}
\bm{s}^{t+1}&=\bm{s}^t-\kappa_{\bm s}^t\nabla_{\bm{s}}C(\bm{\theta}^t,\bm{s}^t), \\
\label{equ:gradientds}
\bm{\theta}^{t+1}&=\bm{\theta}^t-\kappa_{\bm \theta}^t\nabla_{\bm{\theta}}C(\bm{\theta}^t,\bm{s}^{t+1}),
\end{align}
where $ t $ denotes the iteration index, $ \{\kappa_{\bm s}^t,\kappa_{\bm \theta}^t\} $ denote the step size of $ \{\bm{\theta},\bm{s}\} $ in the $t$-th repetition, determined via Armijo line search \cite{boumal2015low}, and $ \nabla_{\cdot} $ denotes the gradient operator w.r.t. $ \cdot $. 


We refer to the methods above as direction finding in the partly calibrated model using \ac{bss} and \ac{mf} (\ac{bss}-\ac{mf}) or \ac{nls} (\ac{bss}-\ac{nls}), respectively. 
The framework of the two algorithms are summarized in Algorithm \ref{alg:PC-BSS-NLS} together, where they differ in step 3). 

\begin{algorithm}[!hbpt]
\caption{\ac{bss}-\ac{mf} / \ac{bss}-\ac{nls}}
\label{alg:PC-BSS-NLS}
{\bf Input:}\ The received signals $ \bm{X} $, the intra-subarray displacement $ \bm{\eta} $ and the number of sources $ L $.
\begin{algorithmic}
\setlength{\baselineskip}{15pt}
\State 1) Apply the \ac{jade} algorithm to \eqref{equ:XBSN} to obtain $ \hat{\bm{S}} $ . 
\State 2) Estimate the phase offsets with \eqref{equ:phaseest}.
\State 3) [\ac{bss}-\ac{mf}]: Estimate directions with \eqref{equ:MF}, or \\
\quad\ [\ac{bss}-\ac{nls}]: Estimate directions by solving \eqref{equ:directionfinding}.
\end{algorithmic}
{\bf Output:}\ Direction estimation $ \hat{\theta}_l $ for $ l=1,\dots,L $.
\end{algorithm}

\section{Resolution and orthogonality}
\label{sec:performanceyanalysis}

In this section, we discuss the angular resolution of the partly and fully calibrated arrays from the aspect of orthogonality in Section \ref{sec:orthogonality}, respectively. 
Since a rigorous analysis of angular resolution, such as \cite{Smith_resolution05}, is usually complex, here we give a heuristic explanation instead and leave the theoretical analysis to future work.

\textbf{Angular resolution of partly calibrated arrays.} As discussed in Proposition \ref{prop:prop2}, when the intervals between sources are greater than $\Delta$, i.e., $\sin\theta_i-\sin\theta_{j} \geqslant \Delta$, $|R_{i,j}|$ is small.
In this case, the phase offsets $ \phi_k(\theta_l,\xi) $ can be well recovered by the algorithms exploiting orthogonality, yielding high-resolution direction finding. 
However, when the intervals between sources are less than $\Delta$, i.e., $\sin\theta_i-\sin\theta_{j}< \Delta$, there is no obvious orthogonality between the source signals. 
This suggests that the achievable angular resolution of our proposed algorithms is on the order of $ \Delta=\lambda/D $, inversely proportional to the whole array aperture $ D $. 
This will be verified by simulation results in Section \ref{sec:simulation}. 

\textbf{Angular resolution of fully calibrated arrays.} We also compare $|R_{i,j}|$ of partly calibrated arrays with the counterpart of fully calibrated arrays.
Particularly, the angular correlation coefficient \cite{asuero2006correlation}, which can be used to indicate the angular resolution, is defined as
\begin{equation}
\label{equ:correlation}
\mathcal{G}_{i,j}=\frac{\bm{g}_j^H\bm{g}_i}{\Vert\bm{g}_i\Vert\Vert\bm{g}_j\Vert},
\end{equation}
where $ [\bm{g}_{\cdot}]_n=e^{\jmath\frac{2\pi}{\lambda}\zeta_n\sin\theta_{\cdot}} $, $ \cdot $ represents $ i $ or $ j $, and $ \zeta_n $ denotes the displacement of the $ n $-th sensor relative to the 1-st sensor in the whole array, $ n=1,\dots,K\bar{M} $.
Similar to Proposition \ref{prop:prop2}, when $ \xi_k \sim \mathcal{U}\left[0,D\right] $, $ |\mathcal{G}_{i,j}| $ is given by Proposition \ref{prop:prop3}.
\begin{prop}
\label{prop:prop3}
When $ \xi_k \sim \mathcal{U}\left[0,D\right]$ for $k=1,\dots,K $, we have
\begin{align}
\label{equ:prop3rand}
\left|\mathbb{E}[\mathcal{G}_{i,j}]\right|&=|\mathcal{M}_{i,j}|\cdot
\left|\frac{\sin\rho_{i,j}}{\rho_{i,j}}\right|, \\
\mathbb{E}\left[|\mathcal{G}_{i,j}|^2\right]&=|\mathcal{M}_{i,j}|^2\left(\frac{1}{K}+\left(1-\frac{1}{K}\right)\left|\frac{\sin \rho_{i,j}}{\rho_{i,j}}\right|^2\right),
\end{align}
where $ \mathcal{M}_{i,j}=\frac{\sin(\bar{M}\varphi_{i,j})}{\bar{M}\sin \varphi_{i,j}} $, $ \varphi_{i,j}=\frac{\pi d }{\lambda}(\sin\theta_i-\sin\theta_j) $ and $ \rho_{i,j}=\frac{\pi D}{\lambda}(\sin\theta_i-\sin\theta_{j}) $.
\end{prop}
\begin{proof}
See Appendix~\ref{app:prop3}.
\end{proof}

From Proposition \ref{prop:prop3}, we take $ \left|\mathbb{E}[\mathcal{G}_{i,j}]\right| $ in \eqref{equ:prop3rand} for example and find that it has two parts, which are derived from the intra-subarray and inter-subarray parts of  the steer vector $ \bm{g} $.
The first part is $ |\mathcal{M}_{i,j}| $, which is related to the subarray aperture.
The second part is the same as \eqref{equ:ERs} in Proposition \ref{prop:prop2}.
Since we assume the same intra-subarray displacements $ \bar{\eta} $ for all the subarrays, this part of $ \left|\mathbb{E}[\mathcal{G}_{i,j}]\right| $ can be extracted and calculated as $ \mathcal{M}_{i,j} $, and the inter-subarray part is then calculated as \eqref{equ:ERs}.

Based on the orthogonality in Proposition \ref{prop:prop2} and Proposition \ref{prop:prop3}, we heuristically see that the fully and partly calibrated arrays have the same angular resolution.
Particularly, since the first zero point can be used to indicate the angular resolution as shown in Fig.~\ref{fig:prop2}, we compare the first zero points of \eqref{equ:ERs} and \eqref{equ:prop3rand}.
In \eqref{equ:prop3rand}, the first zero points of $ |\mathcal{M}_{i,j}| $ and $ \left|\frac{\sin\rho_{i,j}}{\rho_{i,j}}\right| $ are $ \Delta_1=\frac{\lambda}{\bar{M}d} $ and $ \Delta_2=\frac{\lambda}{D} $, respectively.
Since $ \bar{M}d\ll D $, we have $ \Delta_2\ll\Delta_1 $, i.e., the first zero point of $ |\mathcal{G}_{i,j}| $ in \eqref{equ:prop3rand} is $ \frac{\lambda}{D} $, which is the same as $ |R_{i,j}| $ in \eqref{equ:ERs}.
This implies that the angular resolution of the partly calibrated arrays is the same as the fully calibrated arrays'.
A rigorous proof of this conclusion is left to future work.

\section{Experimental results}
\label{sec:simulation}


In this section, we first present two simulation results: The first verifies the ability of phase offset estimation by \ac{jade} in Subsection \ref{subsec:theoryverify}. 
The second compares the performance of our proposed algorithm with existing algorithms and the \ac{crb} \cite{RARE02} in Subsection \ref{subsec:methodcomparison}.
Then, we carry out hardware experiments to further verify the feasibility of our algorithms in Subsection \ref{subsec:experiment}.

\subsection{Phase offset estimation by JADE}
\label{subsec:theoryverify}

Here, we show the performance of phase offset estimation with \ac{jade} by comparing the true $ |R_{i,j}| $ with the estimated one by \ac{jade}.
This is because more accurate phase offset estimation yields better recovery of $ |R_{i,j}| $.
We also verify that higher level of orthogonality likely yields better performance of \ac{jade}.

Particularly, we set $ L=2 $, $ s_1=s_2=1 $, $ \bar{M}=10 $, $ K=10 $, $ D=450\lambda $, $ d=\frac{\lambda}{2} $ and fix $ \theta_1=1.2^{\circ} $.
With $ \sin\theta_2-\sin\theta_1 $ as the variable, we first calculate $ \left| R_{2,1} \right| $ of \eqref{equ:covelements} in the equidistant and random distributed arrays.
Then, we substitute the phase offset estimation as \eqref{equ:phaseest} by \ac{jade} to \eqref{equ:covelements} and calculate the corresponding $ \left| R_{2,1} \right| $ as $ \left| \hat{R}_{2,1} \right|=\left|\frac{1}{K} \sum_{k=1}^K \hat{\phi}_{2,k}\hat{\phi}_{1,k}^*\right| $.
The simulation results are shown in Fig.~\ref{fig:func_angle}. 

\begin{figure}[!htbp]
\centering 
\subfigure[Equidistant case]{
\label{subfigure:orthonow1}
\includegraphics[width=1.6in]{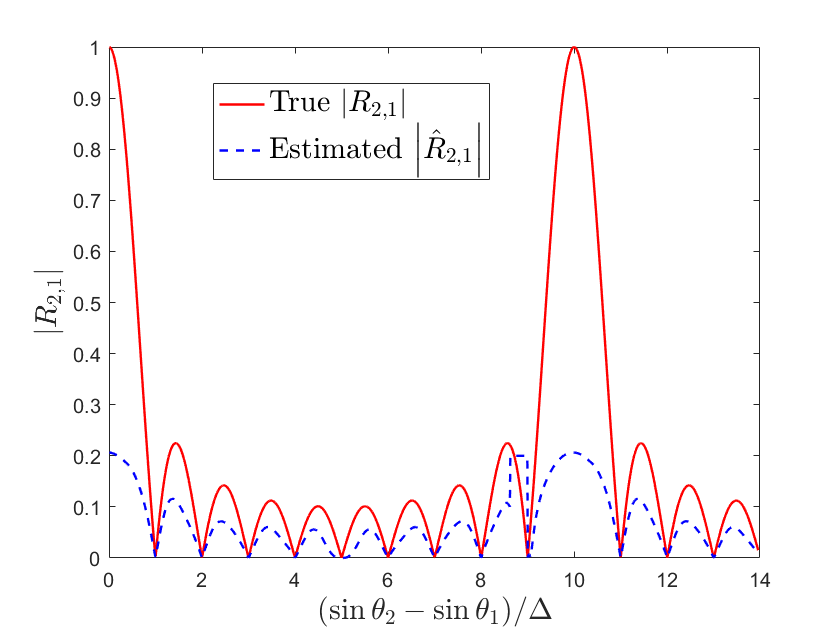}}
\hspace{0in}
\subfigure[Random case]{
\label{subfigure:orthonow2}
\includegraphics[width=1.6in]{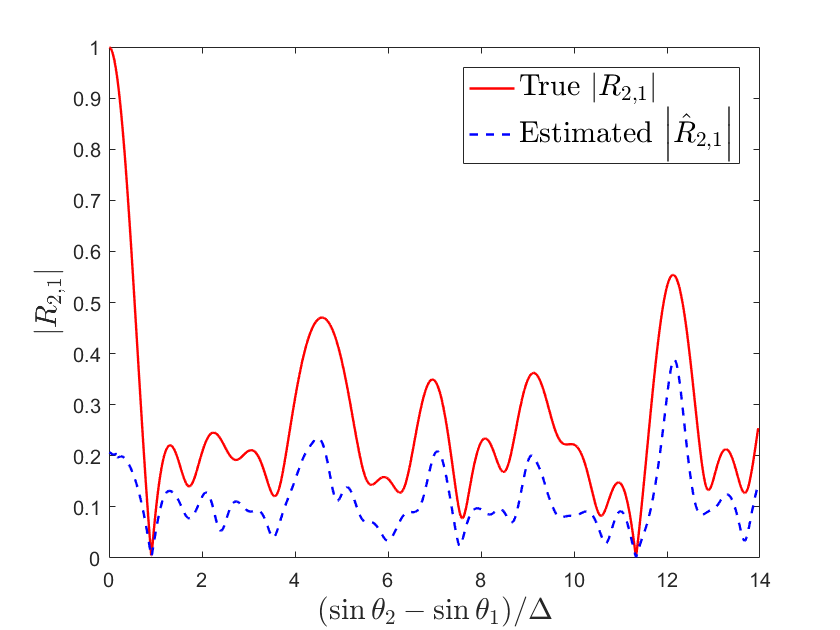}}
\caption{$ |R_{2,1}| $ and $ |\hat{R}_{2,1}| $ w.r.t. $ (\sin\theta_2-\sin\theta_1)/\Delta $ in the equidistant and random cases, where $ \Delta $ denotes angular resolution.}
\label{fig:func_angle} 
\end{figure}

From Fig.~\ref{fig:func_angle}, first we find that the estimation $ |\hat{R}_{2,1}| $ is close to the truth $ |R_{2,1}| $ when $ |R_{2,1}| $ is small, yielding the feasibility of phase recovery in partly calibrated arrays using \ac{jade}.
Second, larger $ |R_{2,1}| $ usually corresponds to a bigger gap between $ |\hat{R}_{2,1}| $ and $ |R_{2,1}| $ and vice versus, which verifies that high level of orthogonality (smaller $ |R_{2,1}| $) likely means better performance of \ac{jade} (smaller gap between $ |\hat{R}_{2,1}| $ and $ |R_{2,1}| $).
Third, $ |R_{2,1}| $ in the equidistant case has apparent grating lobes at $ (\sin\theta_2-\sin\theta_1)/\Delta=10 $, which do not appear in the random case due to the randomness of the displacement.

\subsection{Method comparison}
\label{subsec:methodcomparison}

In this subsection, we compare our proposed algorithms, \ac{bss}-\ac{mf} and \ac{bss}-\ac{nls}, with the root-RARE \cite{RARE02}, spectral-RARE \cite{exRARE04}, COBRAS \cite{COBRAS18}, and group sparse \cite{huang2010benefit} methods in the scenarios with only a single snapshot.
We compare both the angular resolution and direction estimation accuracy of our proposed techniques with other methods and the corresponding \ac{crb}s.
The simulations are carried out under various \ac{snr}s and angle differences. 

The \ac{crb}s for partly calibrated arrays, denoted by PC-\ac{crb} \cite{RARE02}, is both given as the benchmarks.
We use the \ac{rmse} of $ \bm{\theta} $ to indicate the estimation performance of directions. 
We carry out $ T_m $ Monte Carlo trials and denote the \ac{rmse} of $ \bm{\theta} $ by
\begin{equation}
\label{equ:rmse}
{\rm RMSE}(\bm{\theta})=\sqrt{\frac{1}{T_m}\sum_{t=1}^{T_m} \left\Vert \hat{\bm{\theta}}_t-\bm{\theta}^* \right\Vert_2^2},
\end{equation}
where $ \hat{\bm{\theta}}_t $ is the direction estimate in the $ t $-th trial and $ \bm{\theta}^* $ is the true value.

We consider that there are $ L=2 $ emitters in the directions of $ \bm{\theta} $, transmitting signals with wavelength $ \lambda $.
There are $ K=10 $ uniform linear subarrays receiving the signals and each subarray has $ \bar{M} $ array elements.
The intra-subarray and inter-subarray displacements are respectively set as $ \bar{\eta}_m=(m-1)\lambda/2 $ for $ m=1,\dots,\bar{M} $ and $ \xi_k=\frac{(k-1)D}{K-1} $ for $ k=1,\dots,K $, where the whole array aperture is $ D=450\lambda $.
In this case, the angular resolution of a single subarray and the whole array are $ \frac{\lambda}{(\bar{M}-1)d\cos\theta}\approx 12.73^{\circ} $ and $ \frac{\lambda}{D\cos\theta}\approx 0.13^{\circ} $, respectively.
The received signal amplitudes are set as $ \bm{s}=[e^{\jmath\frac{\pi}{5}},3e^{\jmath\frac{3\pi}{5}}]^T $.
In the scenario, there are additive noises $ \bm{n}_k $ for $ k=1,\dots,K $, which are i.i.d. white Gaussian with mean $ \bm{0} $ and variance $ \sigma^2\bm{I} $.
Here \ac{snr} is defined as $ {1}/{\sigma^2} $. 

\subsubsection{Performance w.r.t. SNR}

Here we show the performance of algorithms w.r.t. \ac{snr} in the angular resolution and estimation accuracy.
Particularly, we set $ \bar{M}=10 $ and \ac{snr} $=20$ dB.
Since the angular resolution of the subarray and whole array are $ 12.73^{\circ} $ and $ 1.3^{\circ} $ respectively, we consider two cases that $ \bm{\theta}=[1.2^{\circ},14.2^{\circ}]^T $ and $ [1.2^{\circ},1.4^{\circ}]^T $.
The former can be distinguished by the subarray, and the latter can be distinguished by the whole array, but not by the subarray.

In the group sparse algorithm, we set the grids of $ \bm{\theta} $ as $ \{1,1.1,\dots,15\} $ for $ \bm{\theta}=[1.2^{\circ},14.2^{\circ}]^T $ and $ \{1,1.01,\dots,2\} $ for $ [1.2^{\circ},1.4^{\circ}]^T $, such that the true values are on the grids, and hence there is no grid mismatch problem \cite{5710590}. 
For COBRAS, we set the grids of $ \bm{\theta} $ as $ \{0.2,1.2,\dots,15.2\} $ for $ \bm{\theta}=[1.2^{\circ},14.2^{\circ}]^T $ and $ \{1,1.1,\dots,2\} $ for $ [1.2^{\circ},1.4^{\circ}]^T $, since more grids lead to much higher computation for existing solvers, such as cvx \cite{grant2014cvx}.
For example, the direction finding results of a single trial are shown in Fig.~\ref{fig:com}.
Then, we change \ac{snr} and carry out $ T_m=300 $ trials for each \ac{snr}.
The \ac{rmse}s of \ac{bss}-\ac{nls}, \ac{bss}-\ac{mf}, root-RARE, spectral-RARE, COBRAS and the group sparse (GS) algorithms, as well as PC-\ac{crb}, are shown in Fig.~\ref{fig:RMSE_SNR} with a logarithmic coordinate.

\begin{figure}[!htbp]
\centering 
\subfigure[\text{$ \bm{\theta}=[1.2^{\circ},14.2^{\circ}]^T $}]{
\label{subfigure:com1}
\includegraphics[width=1.6in]{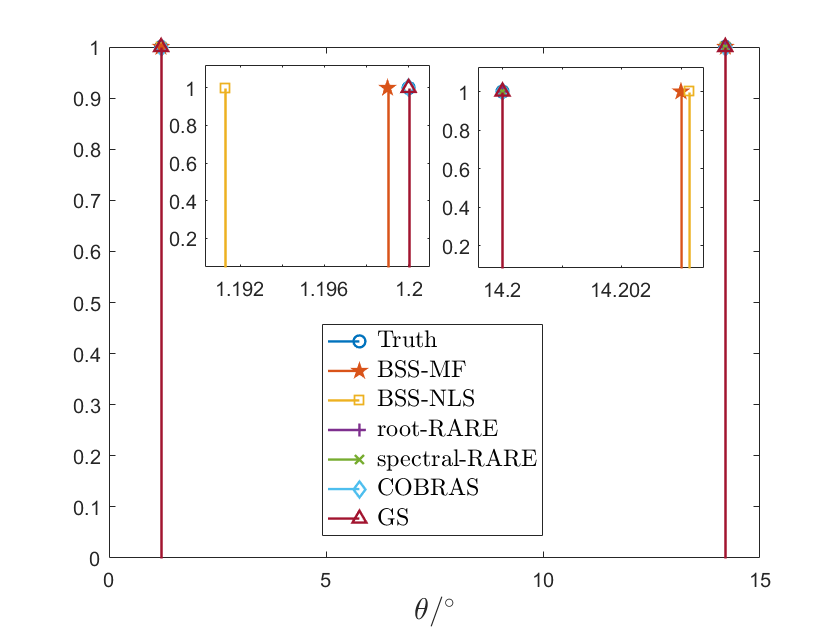}}
\hspace{0in}
\subfigure[\text{$ \bm{\theta}=[1.2^{\circ}, 1.4^{\circ}]^T $}]{
\label{subfigure:com2}
\includegraphics[width=1.6in]{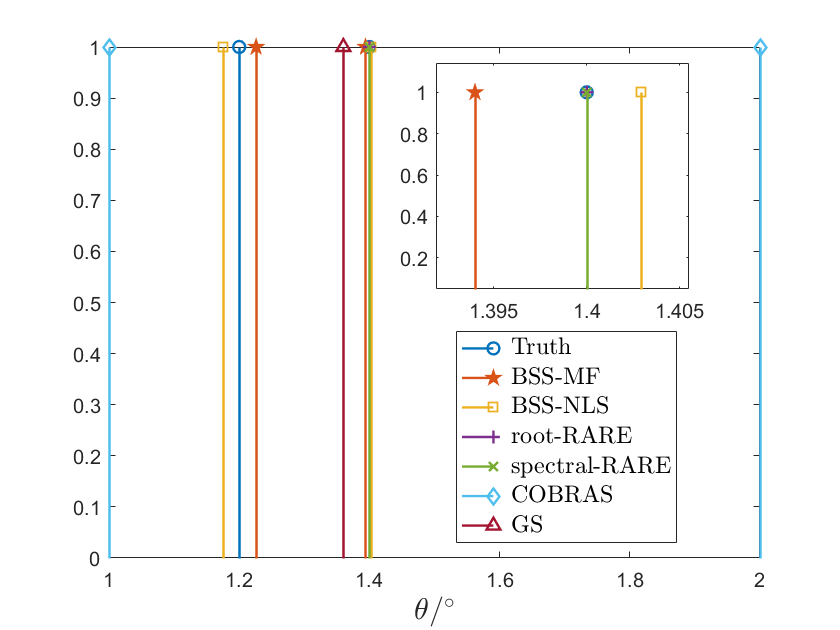}}
\caption{Comparison in the angular resolution.}
\label{fig:com} 
\end{figure}

\begin{figure}[!htbp]
\centering
\subfigure[\text{$ \bm{\theta}=[1.2^{\circ},14.2^{\circ}]^T $}]{
\label{subfigure:SNRa}
\includegraphics[width=1.6in]{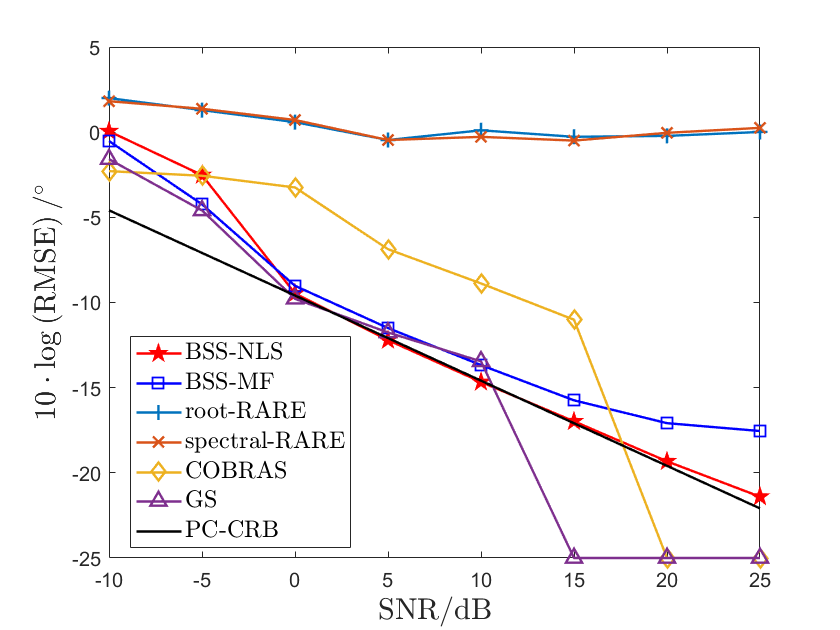}}
\hspace{0in}
\subfigure[\text{$ \bm{\theta}=[1.2^{\circ}, 1.4^{\circ}]^T $}]{
\label{subfigure:SNRb}
\includegraphics[width=1.6in]{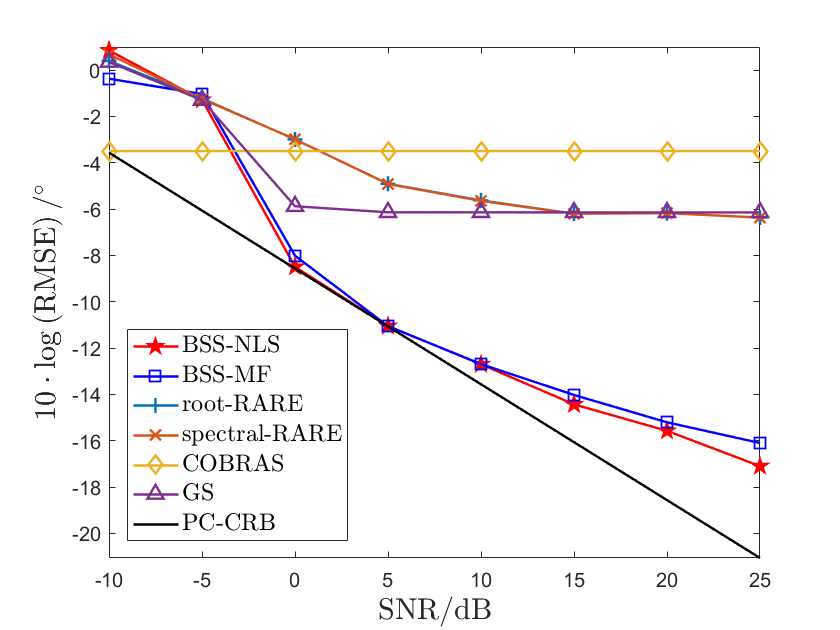}}
\caption{RMSE w.r.t. \ac{snr}.}
\label{fig:RMSE_SNR}
\end{figure}

From Fig.~\ref{fig:com}, we first find that root-RARE and spectral-RARE algorithms only find the source at $ \theta=14.2^{\circ}/1.4^{\circ} $ and fail to identify the two sources in both cases.
This is because only a single snapshot is available in our scenarios, but these algorithms require enough snapshots for subspace separation. 
Then in Fig.~\ref{subfigure:com1}, we find that the group sparse and COBRAS algorithms separate the two sources and achieve good estimation when $ \bm{\theta}=[1.2^{\circ},14.2^{\circ}]^T $.
But they can only identify one emitter when $ \bm{\theta}=[1.2^{\circ},1.4^{\circ}]^T $ as shown in Fig.~\ref{subfigure:com2}.
This indicates that these sparse recovery algorithms cannot achieve the angular resolution corresponding to the whole array aperture.
However, our proposed algorithms, \ac{bss}-\ac{mf} and \ac{bss}-\ac{nls}, are shown to identify the two emitters and estimate the directions accurately, which indicates the ability of exploiting the whole array aperture and achieving high angular resolution. 

In Fig.~\ref{fig:RMSE_SNR}, the statistical results match with the results of a single trial.
The root-RARE and spectral-RARE algorithms have poor performance in both cases.
The group sparse and COBRAS algorithms estimate $ \bm{\theta} $ well when $ \bm{\theta}=[1.2^{\circ},14.2^{\circ}]^T $.
For high \ac{snr} (\ac{snr} $\geqslant 20$dB) in Fig.~\ref{subfigure:SNRa}, these sparse recovery algorithms almost exactly estimate the directions, because the true values are assumed on the grids (Due to space limit, we use -25dB to denote the completely exact estimation).
But when $ \bm{\theta}=[1.2^{\circ},1.4^{\circ}]^T $, they fail in high \ac{snr}s.
This further shows that these sparse recovery algorithms cannot exploit the whole array aperture.
However, the \ac{rmse}s of our proposed algorithms are close to PC-\ac{crb} and outperform other algorithms when $ \bm{\theta}=[1.2^{\circ},1.4^{\circ}]^T $.
This further verifies that our algorithms can achieve high angular resolution, inversely proportional to the whole array aperture.

\subsubsection{Performance w.r.t. the angle difference}

Here we show the performance of the algorithms w.r.t. the angle difference to have a detailed analysis on the angular resolution.
Particularly, we set $ \bar{M}=30 $, \ac{snr} $=20$ dB, fix $ \theta_1=1.2^{\circ} $, and change $ \theta_2 $ from $ 1.22^{\circ} $ to $ 1.60^{\circ} $. 
To compare with the feature of $ |R_{2,1}| $ in Fig.~\ref{fig:func_angle}, we consider the \ac{rmse}s of algorithms w.r.t. $ (\sin\theta_2-\sin\theta_1)/\Delta $.
The simulation results are shown in Fig.~\ref{fig:RMSE_angle} with a logarithmic coordinate.

\begin{figure}[!htbp]
\centering
\includegraphics[width=3in]{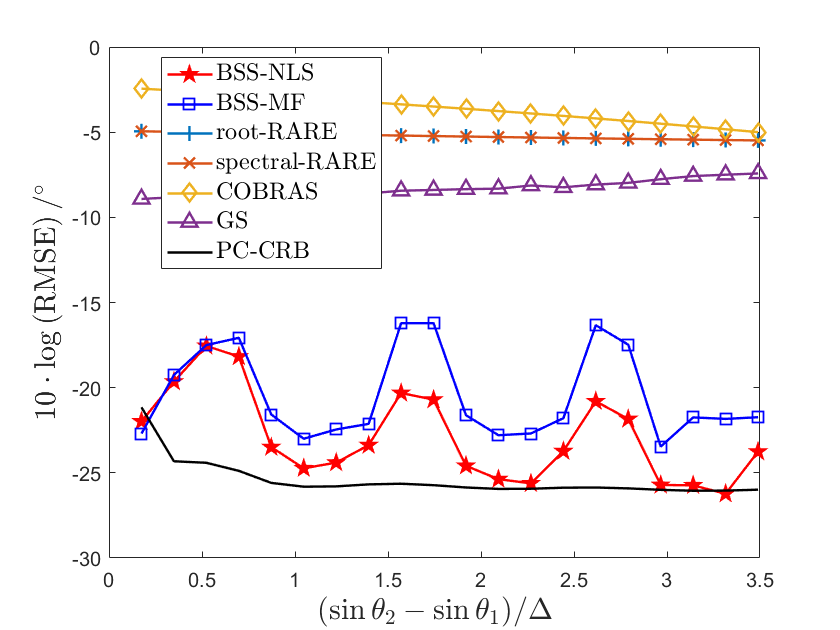}
\caption{RMSE w.r.t. the angle difference.}
\label{fig:RMSE_angle} 
\end{figure}

In Fig.~\ref{fig:RMSE_angle}, first we find that the \ac{rmse}s of our proposed algorithm are smaller than other algorithms and close to PC-\ac{crb} when $ (\sin\theta_2-\sin\theta_1)/\Delta $ is an integer, yielding better estimation accuracy. 
Second, it is worth noting that the \ac{rmse}s of our proposed algorithms periodically get worse w.r.t. the angle difference.
We use \ac{jade} for phase offset estimation and verify that smaller $ |R_{i,j}| $ likely yields better \ac{jade} performance in Subsection \ref{subsec:theoryverify}. 
The increase of the angle difference changes $ |R_{i,j}| $ periodically, which affects the \ac{jade} performance in phase offset estimation, and eventually leads to periodic variations in the direction finding performance.
Particularly, the curves of \ac{bss}-\ac{nls} and \ac{bss}-\ac{mf} in Fig.~\ref{fig:RMSE_angle} approximate to the minimums when $ (\sin\theta_2-\sin\theta_1)/\Delta=1,2,3,\dots $, which is consistent with the minimums of $ |R_{2,1}| $ in Fig.~\ref{subfigure:orthonow1}.
These phenomena support that smaller $ |R_{i,j}| $ corresponds to higher level of orthogonality, yielding better performance. 

\subsection{Experiment results}
\label{subsec:experiment}


We consider a scenario of far-field targets impinging narrow-band signals onto the partly calibrated array.
This scenario is achieved based on a \ac{vna}.
Particularly, the two ports of the \ac{vna} are connected to the transmit (Tx) and receive (Rx) antennas, respectively.
We view the transmit antennas as the sources.
We construct the received signals of the partly calibrated array by sequentially moving the receive antenna to the position of each element and then sampling the received signals.
The received signals of the two sources, T1 and T2, are sampled separately, denoted by $ \bm{t}_1\in\mathbb{C}^{M\times1} $ and $ \bm{t}_2\in\mathbb{C}^{M\times1} $.
We view $ \bm{t}_1+\bm{t}_2 $ as the received signals of the two sources.
The geometry is shown in Fig.~\ref{fig:exp1} and the practical scenario is shown in Fig.~\ref{fig:exp2}.

\begin{figure}[!htbp]
\centering
\includegraphics[width=3in]{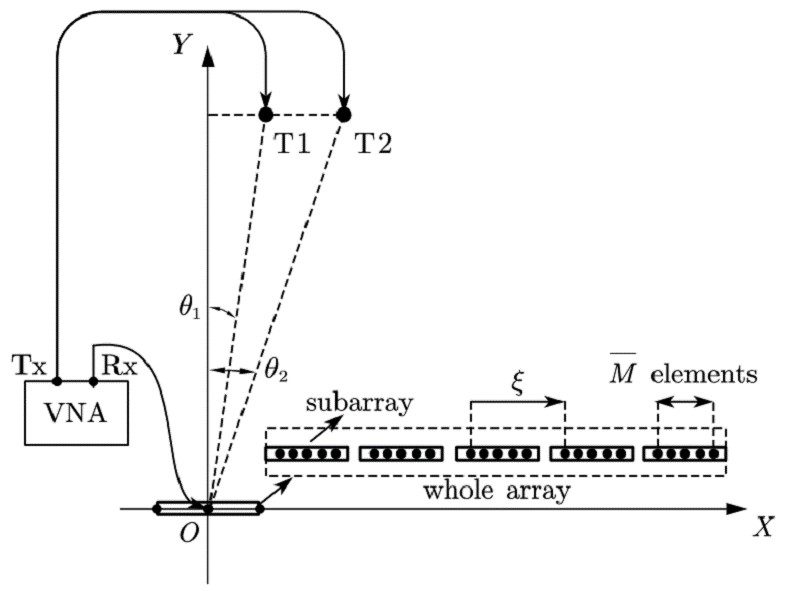}
\caption{The geometry of the experiment.}
\label{fig:exp1} 
\end{figure}

\begin{figure}[!htbp]
\centering
\includegraphics[width=3in]{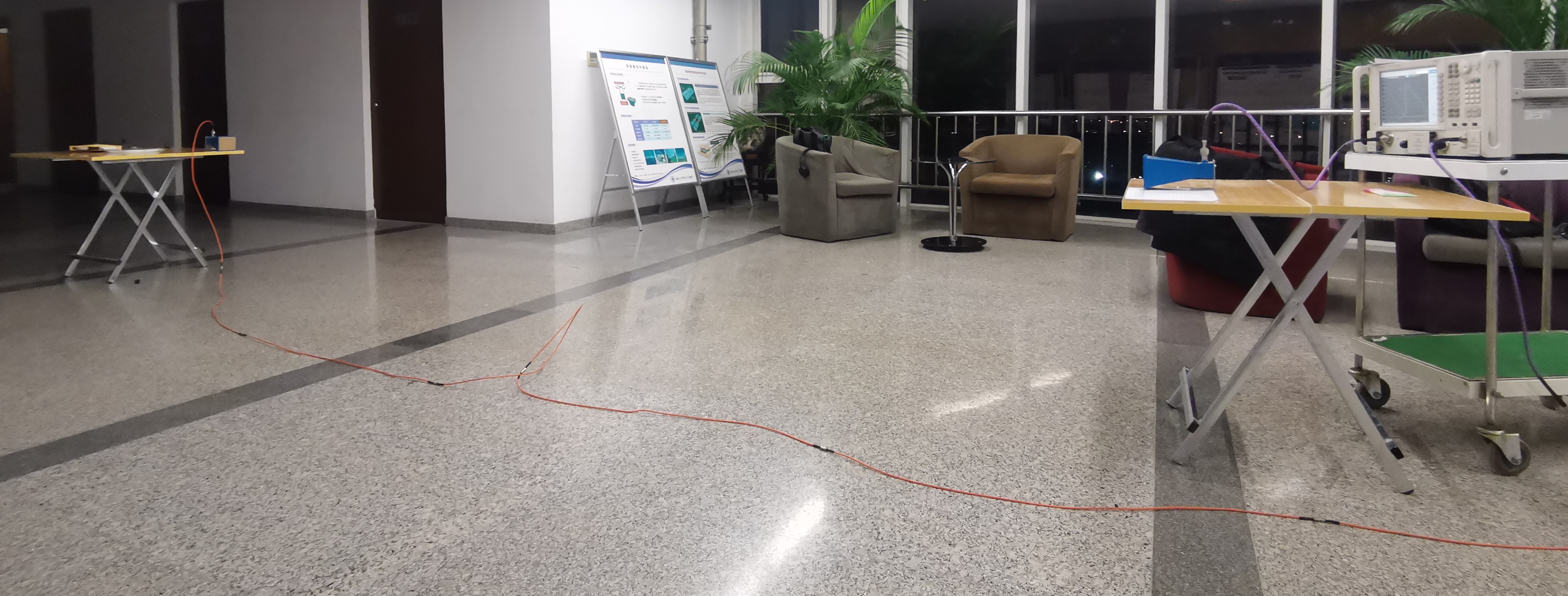}
\caption{Photo of the experiment scenario.}
\label{fig:exp2} 
\end{figure}

The parameter settings of this experiment are shown as follows:
The frequency of the transmitted signals is $ f_c=15 $ GHz, and the wavelength is $ \lambda=2 $ cm.
There are $ K=5 $ uniform linear subarrays and each has $ \bar{M}=5 $ array elements with $ d=1 $ cm.
Particularly, we construct a planar $XOY$ coordinate system and the unit is centimeter (cm).
The elements are located on the X axis and the coordinates are $ (-12,0),(-11,0),\dots,(12,0) $.
Therefore, the angular resolution of a single subarray and the whole array are $ \frac{\lambda}{(\bar{M}-1)d\cos\theta}\approx 28.6^{\circ} $ and $ \frac{\lambda}{D\cos\theta}\approx 5.7^{\circ} $, respectively.
The positions of the sources are $ (5.46,500) $ and $ (58.27,500) $.
We take the array center $ \bm{O}(0,0) $ as the reference point and the directions of the sources are $ \bm{\theta}=[0.63^{\circ},6.65^{\circ}]^T $.
The array positions are exactly measured using a grid paper with $ 1 $ cm interval, but the source positions are not because they are far from the receiver and measurement errors are inevitable.  
To obtain the true directions $ \bm{\theta} $ in practice, we use the beamforming results of $ \bm{t}_1 $ and $ \bm{t}_2 $ with the exact $ \bm{\xi} $ as the ground truth instead. 

Under the above settings, we have the received signals $ \bm{t}_1 $ and $ \bm{t}_2 $.
Recall that the fully and partly calibrated model assume exact and biased inter-subarray displacements $ \bm{\xi} $, respectively.
The exact $ \bm{\xi} $ is set above and the biased $ \hat{\bm{\xi}} $ is constructed by adding errors to the exact $ \bm{\xi} $ artificially.
Particularly, we take $ \lambda/2 $ as unit and add a Gaussian error to exact $ \bm{\xi} $ to construct biased $ \hat{\bm{\xi}} $ as $ \hat{\bm{\xi}}=\bm{\xi}+\Delta\bm{\xi} $, $ \Delta\bm{\xi}\sim N(\bm{0},\sigma_{\xi}^2\bm{I}) $. 
Here we set $ 10{\rm log}(1/\sigma_{\xi}^2)=10 $ dB.
Then, we calculate the beamforming results of $ \bm{t}_1+\bm{t}_2 $ with biased $ \hat{\bm{\xi}} $ and compare it with the estimation of \ac{bss}-\ac{mf}, \ac{bss}-\ac{nls} and the group sparse algorithm (We do not consider the root-RARE and spectral-RARE algorithms here since they fail in the single-snapshot cases). 
The comparison is shown in Fig.~\ref{fig:expsimu}.

\begin{figure}[!htbp]
\centering
\subfigure[Beamforming  (BF) results  with  $ \bm{\xi} $]{\label{subfigure:exp1}
\includegraphics[width=1.6in]{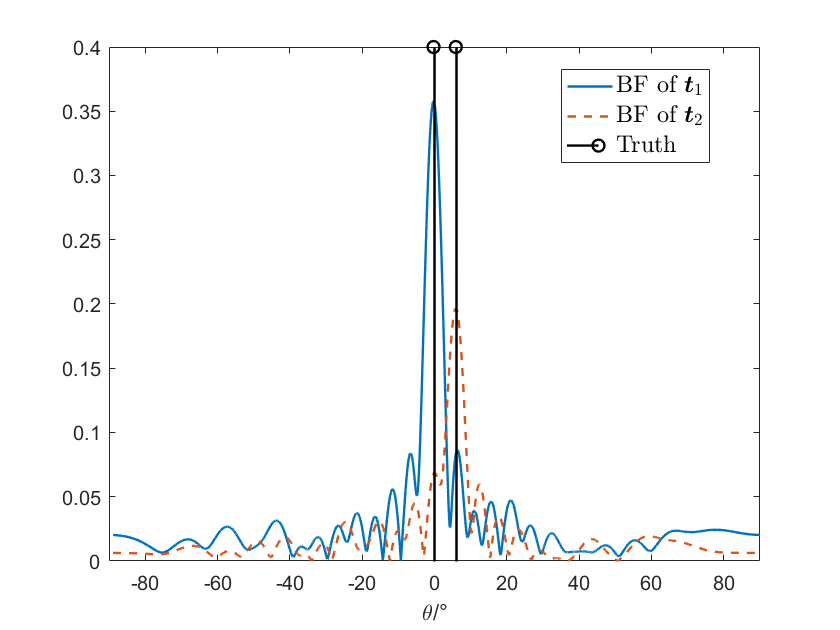}}
\hspace{0in}
\subfigure[\ac{doa} estimates with biased $ \hat{\bm{\xi}} $]{
\label{subfigure:exp4}
\includegraphics[width=1.6in]{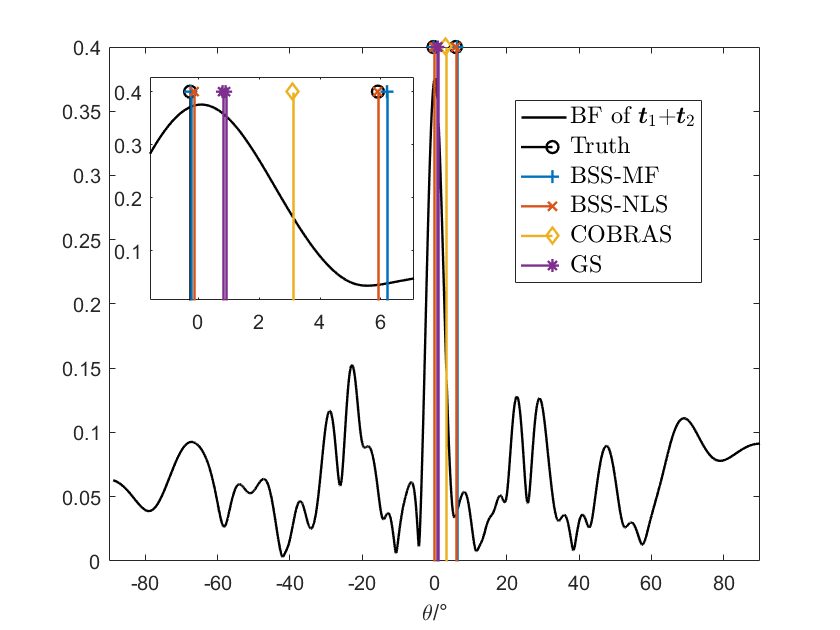}}
\caption{Comparison between beamforming and other  algorithms on experimental data $\bm t_1$, $\bm t_2$ and $\bm t_1 + \bm t_2$.}
\label{fig:expsimu}
\end{figure}

From Fig.~\ref{subfigure:exp1}, we find that the main lobes of both sources are clear and the side lobes are low.
The peaks of the main lobes are corresponding to $ -0.26^{\circ} $ and $ 5.90^{\circ} $, which are close to our settings.
These phenomenons imply that we get high-quality real measured data.

From Fig.~\ref{subfigure:exp4}, when the inter-subarray displacements are erroneous, we find that the beamforming results have many high side lobes and it is hard to identify both sources.
This phenomenon illustrates that the subarray position errors in the partly calibrated model have a severe effect on the direction finding performance.
However, our proposed algorithms are shown to identify the two sources, with the estimates close to the truth, outperforming the group sparse and COBRAS techniques.

\section{Conclusion}
\label{sec:conclusion}

In this paper, we consider achieving high angular resolution using distributed arrays in the presence of array displacement errors. 
We propose a novel algorithm with only a single snapshot, yielding less data transmission burden and processing delay compared with existing works.
This is achieved by exploiting the orthogonality between source signals, using \ac{bss}.
Based on \ac{bss} ideas, we propose a direction finding algorithm, and discuss the relationship between orthogonality and angular resolution.
Simulation and experiment results both verify the feasibility of our algorithms in achieving high angular resolution, inversely proportional to the whole array aperture.

\appendices

\section{Proof of Proposition \ref{prop:prop2}}
\label{app:prop2}

Here we prove Proposition \ref{prop:prop2}, i.e., calculate $ \left|\mathbb{E}\left[R_{i,j}\right]\right| $ and $ \mathbb{E}\left[\left|R_{i,j}\right|^2\right] $ when $ \xi_k\sim \mathcal{U}\left[0,D\right] $ for $ k=1,\dots,K $.
We denote $ \rho_{i,j}=\frac{\pi D}{\lambda}(\sin\theta_i-\sin\theta_{j}) $.

For $ \left|\mathbb{E}\left[R_{i,j}\right]\right| $, based on \eqref{equ:covelements}, we have
\begin{align}
\label{equ:exprand}
\left|\mathbb{E}\left[R_{i,j}\right]\right|&=\frac{1}{K}\left|\mathbb{E}\left[\sum_{k=1}^{K} e^{\jmath2\rho_{i,j}\xi_k/D}\right]\right| \nonumber \\
&=\frac{1}{K}\left|\sum_{k=1}^{K}\mathbb{E}\left[ e^{\jmath2\rho_{i,j}\xi_k/D}\right]\right|.
\end{align}
Since $ \xi_k \sim \mathcal{U}\left[0,D\right]$, $k=1,\dots,K $, the $ K $ expectations in \eqref{equ:exprand} are the same.
Therefore, we rewrite \eqref{equ:exprand} as
\begin{align}
\label{equ:exprand2}
\left|\mathbb{E}\left[R_{i,j}\right]\right|&= \frac{1}{K}\cdot K\cdot\left|\mathbb{E}\left[e^{\jmath2\rho_{i,j}\xi_k/D}\right]\right| \nonumber \\
&=\left|\frac{1}{D}\int_{0}^{D} e^{\jmath2\rho_{i,j}\xi/D}d\xi\right| \nonumber \\
&=\left|\frac{\sin\rho_{i,j}}{\rho_{i,j}}\right|.
\end{align}

For $ \mathbb{E}\left[\left|R_{i,j}\right|^2\right] $, based on \eqref{equ:covelements}, we have
\begin{align}
\label{equ:ERprop}
\mathbb{E}\left[\left|R_{i,j}\right|^2\right]&=\frac{1}{K^2}\mathbb{E}\left[\left|\sum_{k=1}^{K} e^{\jmath2\rho_{i,j}\xi_k/D}\right|^2\right].
\end{align}
Due to $ |x|^2=x^H x $, we write \eqref{equ:ERprop} as 
\begin{align}
\label{equ:ERprop2}
\mathbb{E}\left[\left|R_{i,j}\right|^2\right]&=\frac{1}{K^2}\mathbb{E}\left[\sum_{p=1}^{K} e^{-\jmath2\rho_{i,j}\xi_p/D}\sum_{q=1}^{K} e^{\jmath2\rho_{i,j}\xi_q/D}\right], \nonumber
\end{align}
where the expectation term further equals
\begin{equation*}
\small
    \mathbb{E}\left[\sum_{p=1}^K\sum_{q=1}^K e^{\jmath2\rho_{i,j}(\xi_q-\xi_p)/D}\right]
=\sum_{p=1}^K\sum_{q=1}^K\mathbb{E}\left[e^{\jmath2\rho_{i,j}(\xi_q-\xi_p)/D}\right].
\end{equation*}
When $ p=q $, $ e^{\jmath2\rho_{i,j}(\xi_q-\xi_p)/D}=1 $, we have
\begin{align}
\label{equ:ERprop3}
\mathbb{E}\left[\left|R_{i,j}\right|^2\right]&=\frac{1}{K}+\frac{1}{K^2}\sum_{p=1}^K\sum_{q=1,q\ne p}^K\mathbb{E}\left[e^{\jmath2\rho_{i,j}(\xi_q-\xi_p)/D}\right].
\end{align}
Since $ \xi_k \sim \mathcal{U}\left[0,D\right]$, $k=1,\dots,K $, the $ K^2-K $ expectations in \eqref{equ:ERprop3} are the same.
Hence, we write \eqref{equ:ERprop3} as
\begin{align}
\label{equ:ERprop4}
\mathbb{E}\left[\left|R_{i,j}\right|^2\right]&=\frac{1}{K}+\frac{K^2-K}{K^2D^2}\int_{0}^{D}\int_{0}^{D}e^{\jmath2\rho_{i,j}(\xi_q-\xi_p)/D} d\xi_pd\xi_q \nonumber \\
&=\frac{1}{K}+\frac{K^2-K}{K^2D^2}\cdot\frac{D^2\sin^2\rho_{i,j}}{\rho_{i,j}^2} \nonumber \\
&=\frac{1}{K}+\left(1-\frac{1}{K}\right)\frac{\sin^2\rho_{i,j}}{\rho_{i,j}^2},
\end{align}
completing the proof.

\section{Proof of Proposition \ref{prop:prop3}}
\label{app:prop3}

Based on the definition of $ \mathcal{G}_{i,j} $, we have
\begin{align}
\label{equ:defX}
\left|\mathcal{G}_{i,j}\right|&=\left|\frac{\bm{g}_i^H\bm{g}_j}{\Vert\bm{g}_i\Vert\Vert\bm{g}_j\Vert}\right|=\frac{1}{K\bar{M}}\left|\sum_{n=1}^{K\bar{M}}e^{\jmath \frac{2\pi}{\lambda}\zeta_n(\sin\theta_i-\sin\theta_j)}\right|.
\end{align}
Note that $ \zeta_n, n=1,\dots,K\bar{M} $ is related to both the intra-subarray $ \bar{\bm{\eta}} $ and inter-subarray displacement $ \bm{\xi} $.
Therefore, we can represent $ \zeta_n $ of the $ m $-th sensor in the $ k $-th subarray as $ \zeta_n=\xi_k+\bar{\eta}_m $.
Then, we rewrite \eqref{equ:defX} as
\begin{align}
\label{equ:Xrewrite}
&\left|\mathcal{G}_{i,j}\right|=\nonumber \\
&\left|\frac{1}{\bar{M}}\sum_{m=1}^{\bar{M}}e^{\jmath \frac{2\pi}{\lambda}\bar{\eta}_m(\sin\theta_i-\sin\theta_j)}\right|\cdot\left|\frac{1}{K}\sum_{k=1}^{K}e^{\jmath \frac{2\pi}{\lambda}\xi_k(\sin\theta_i-\sin\theta_j)}\right|.
\end{align}
Since $ \bar{\bm{\eta}}_m=(m-1)d $ is a constant, we calculate the first summation of \eqref{equ:Xrewrite} as
\begin{align}
\label{equ:firstsum}
|\mathcal{M}_{i,j}|&\equiv\left|\frac{1}{\bar{M}}\sum_{m=1}^{\bar{M}}e^{\jmath \frac{2\pi}{\lambda}\bar{\eta}_m(\sin\theta_i-\sin\theta_j)}\right| \nonumber\\
&=\left|\frac{\sin(\bar{M}\varphi_{i,j})}{\bar{M}\sin \varphi_{i,j}}\right|,
\end{align}
where $ \varphi_{i,j}=\frac{\pi d}{\lambda}(\sin\theta_i-\sin\theta_j) $.

Next we consider the second summation of \eqref{equ:Xrewrite}.
Note that it is the same as the calculation of \eqref{equ:covelements}.
Therefore, based on Proposition \ref{prop:prop2}, we directly have 
\begin{align}
\label{equ:directfromprop2}
\left|\mathbb{E}[\mathcal{G}_{i,j}]\right|&=|\mathcal{M}_{i,j}|\cdot
\left|\frac{\sin\rho_{i,j}}{\rho_{i,j}}\right|, \\
\mathbb{E}\left[|\mathcal{G}_{i,j}|^2\right]&=|\mathcal{M}_{i,j}|^2\left(\frac{1}{K}+\left(1-\frac{1}{K}\right)\left|\frac{\sin \rho_{i,j}}{\rho_{i,j}}\right|^2\right),
\end{align}
where $ \rho_{i,j}=\frac{\pi D}{\lambda}(\sin\theta_i-\sin\theta_{j}) $, completing the proof.

\section{The JADE algorithm}
\label{app:JADEalg}

In \eqref{equ:YCHN}, the \ac{jade} algorithm considers recovering $ \bm{H} $ from $ \bm{Y} $ with unknown $ \bm{C} $.
The main framework of \ac{jade} is:
\begin{itemize}
    \item[(1)] Compute a whitening matrix $ \bm{W}\in\mathbb{C}^{L\times N} $ based on $ \bm{Y} $.
    \item[(2)] Calculate the whitened data $ \bm{Z}=\bm{W}\bm{Y}\in\mathbb{C}^{L\times T_s} $.
    \item[(3)] Find a unitary matrix $ \bm{V}\in\mathbb{C}^{L\times L} $, such that $ f_s(\bm{V}^H\bm{Z}) $ reaches its maximum, where $ f_s(\cdot) $ is a function characterizing the level of independence.
    \item[(4)] Output the estimation of $ \bm{H} $ as $ \hat{\bm{H}}=\bm{V}^H\bm{W}\bm{Y} $.
\end{itemize}

Here we detail the specific procedures of \ac{jade}.
In step (1), to calculate the whitening matrix, \ac{jade} first calculates the sample covariance matrix of $ \bm{Y} $, denoted by $ \hat{\bm R}_{\bm y}\equiv\frac{1}{T_s}\bm{Y}\bm{Y}^H $.
In $ \hat{\bm R}_{\bm y} $, the noise variance $ \hat{\sigma} $ is estimated as the average of the $ N-L $ smallest eigenvalues of $ \hat{\bm{R}}_{\bm{y}} $, given by $ \hat{\sigma}=\frac{1}{N-L}\sum_{l=L+1}^{N}[\hat{\bm{\Sigma}}_{\bm{y}}]_{l,l} $, where $ \hat{\bm{\Sigma}}_{\bm{y}} $ is obtained from the \ac{svd} of $ \hat{\bm{R}}_{\bm{y}} $, i.e., $ \hat{\bm{R}}_{\bm{y}}=\hat{\bm{U}}_{\bm{y}} \hat{\bm{\Sigma}}_{\bm{y}} \hat{\bm{U}}_{\bm{y}}^H $.
Therefore, the white noise covariance is estimated as $ \hat{\sigma}\bm{I} $.
Subtracting the noise covariance from $ \hat{\bm{\Sigma}}_{\bm{y}} $ implies noiseless estimation of $ \bm{\Sigma}_{\bm{y}} $.
Then, the whitening matrix $ \bm{W} $ is estimated as 
\begin{equation}
\label{equ:Wspecific}
\hat{\bm{W}}=[\hat{\bm \Sigma}_{\bm{y}}^T-\hat{\sigma}\bm{I}]_{1:L}^{-\frac{1}{2}}\hat{\bm{U}}_{\bm{y}}^H,
\end{equation}
where $ [\cdot]_{1:L} $ denotes the first $ L $ columns of matrix $ \cdot $.

In step (3), \ac{jade} minimizes the following cost function characterizing independence, given by
\begin{equation}
\label{equ:Vopt}
    f_s(\bm{G})=\sum_{r,p,q=1,\dots,L} \Big|\widehat{\text{Cum}}\left([\bm{G}^T]_r,[\bm{G}^H]_r,[\bm{G}^T]_p,[\bm{G}^H]_q\right)\Big|^2,
\end{equation}
where $ \widehat{\text{Cum}}(\cdot) $ is denoted by
\begin{align}
\label{equ:cumulant}
&\widehat{\text{Cum}}(\bm{\epsilon}_a,\bm{\epsilon}_b,\bm{\epsilon}_c,\bm{\epsilon}_d)= \frac{1}{T_s}(\bm{\epsilon}_a\odot\bm{\epsilon}_b)^T(\bm{\epsilon}_c\odot\bm{\epsilon}_d) \nonumber \\
&-\frac{1}{T_s^2}\left(\bm{\epsilon}_a^T\bm{\epsilon}_b\bm{\epsilon}_c^T\bm{\epsilon}_d+\bm{\epsilon}_a^T\bm{\epsilon}_c\bm{\epsilon}_b^T\bm{\epsilon}_d+\bm{\epsilon}_a^T\bm{\epsilon}_d\bm{\epsilon}_b^T\bm{\epsilon}_c\right),
\end{align}
for $ \bm{\epsilon}_a=[\bm{G}^T]_{a}, \bm{\epsilon}_b=[\bm{G}^H]_{b}, \bm{\epsilon}_c=[\bm{G}^T]_{c}, \bm{\epsilon}_d=[\bm{G}^H]_{d} $ for $ a,b,c,d=1,\dots,L $.
In \cite{cardoso1993blind}, the minimization of \eqref{equ:Vopt} w.r.t. $ \bm{G} $ is proven to be equivalent to a joint approximate diagonalization of some cumulant matrices and solved efficiently.

The \ac{jade} algorithm is summarized as Algorithm \ref{alg:jade}.
As mentioned above, the \ac{jade} algorithm transforms the maximization of \eqref{equ:Vopt} to a joint approximate diagonalization of some cumulant matrices.
The steps (2) and (3) in Algorithm \ref{alg:jade} are to construct the cumulant matrices, and step (4) is the joint approximate diagonalization.
\begin{algorithm}[!hbpt]
\caption{\ac{jade} \cite{cardoso1993blind}}
\label{alg:jade}
{\bf Input:}\ The received signals $ \bm{Y} $ and the number of sources $ L $.
\begin{algorithmic}
\setlength{\baselineskip}{15pt}
\State (1) Calculate $ \hat{\bm{W}} $ as \eqref{equ:Wspecific} and let $ \bm{Z}\equiv\hat{\bm{W}}\bm{Y} $. 
\State (2) Calculate $ \hat{\bm{Q}}_z (\bm Z)$, given by,
\begin{equation}
[\hat{\bm{Q}}_z]_{p,q}=\widehat{\text{Cum}}([\bm{Z}^T]_{a},[\bm{Z}^H]_{b},[\bm{Z}^T]_{c},[\bm{Z}^H]_{d}),
\end{equation}
where $ \widehat{\text{Cum}}(\cdot) $ is denoted by \eqref{equ:cumulant}, $ p=a+(b-1)L $ and $ q=d+(c-1)L $ for $ a,b,c,d=1,\dots,L $. 
\State (3) Let $ \varrho_l=[\bm{\Sigma}_Q]_{l,l} $ and $ \bm{r}_l=[\bm{U}_Q]_{l} $, where $ \bm{\Sigma}_Q,\bm{U}_Q $ are from the \ac{svd} of $ \hat{\bm{Q}}_z $, i.e., $ \hat{\bm{Q}}_z=\bm{U}_Q\bm{\Sigma}_Q\bm{U}_Q^H $ and define $ \bm{R}_l\in\mathbb{C}^{L\times L} $, s.t., $ \text{vec}(\bm{R}_l)=\bm{r}_l $ for $ l=1,\dots,L $.
\State (4) Jointly diagonalize $ \mathcal{R}\equiv\left\{\widetilde{\bm{R}}_l=\varrho_l\bm{R}_l\mid 1\leqslant l\leqslant L \right\} $ by a unitary matrix $ \bm{V} $, given in Algorithm \ref{alg:jointdiag}.
\State (5) Estimate $ \bm{H} $ as $ \hat{\bm{H}}=\bm{V}^H\hat{\bm{W}}\bm{Y} $.
\end{algorithmic} 
{\bf Output:}\ The waveform estimation $ \hat{\bm{H}} $.
\end{algorithm}

\begin{algorithm}[!hbpt]
\caption{Joint approximate diagonalization}
\label{alg:jointdiag}
{\bf Input:}\ The matrix set $ \mathcal{R} $ and the number of sources $ L $.
\begin{algorithmic}
\setlength{\baselineskip}{15pt}
\State \textbf{For} $ m=1,\dots,L-1 $:

\textbf{For} $ n=m+1,\dots,L $:

1) Let $ \bm{O}_{m,n}=[\bm{o}_{m,n}^1,\dots,\bm{o}_{m,n}^L]^T $, where 
\begin{align}
\bm{o}_{m,n}^{l}=\big[&[\tilde{\bm{R}}_l]_{m,m}-[\tilde{\bm{R}}_l]_{n,n},[\tilde{\bm{R}}_l]_{m,n}+[\tilde{\bm{R}}_l]_{n,m},\nonumber \\
&\jmath([\tilde{\bm{R}}_l]_{n,m}-[\tilde{\bm{R}}_l]_{m,n})\big]^T.
\end{align}

2) Let $ \bm{\eta}_{m,n}=[\bm{U}_o]_{1} $, where $ \bm{U}_o $ is from the \ac{svd} as $ {\rm Re}(\bm{O}_{m,n}^H\bm{O}_{m,n})=\bm{U}_o\bm{\Sigma}_o\bm{U}_o^H $. 

3) Calculate $ \alpha_{m,n} $ and $ \beta_{m,n} $ as
\begin{align}
\label{equ:betaphi}
\alpha_{m,n}&=\sqrt{\frac{1+[\bm{\eta}_{m,n}]_{1,1}}{2}}, \nonumber \\
\beta_{m,n}&=\frac{[\bm{\eta}_{m,n}]_{2,1}-\jmath[\bm{\eta}_{m,n}]_{3,1}}{2\alpha_{m,n}}.
\end{align}

4) Define $ \bm{G}_{m,n} $ by 
\begin{align}
\bm{G}_{m,n}=&\bm{I}+\beta_{m,n}\bm{e}_n\bm{e}_m^T-\beta_{m,n}^*\bm{e}_m\bm{e}_n^T \nonumber \\
&+(\alpha_{m,n}-1)(\bm{e}_m\bm{e}_m^T-\bm{e}_n\bm{e}_n^T),
\end{align}
where $ \bm{e}_{\cdot}\in\mathbb{C}^{L\times 1} $ is the unit vector.

5) Update the set $ \mathcal{R} $ as 
\begin{equation}
\label{equ:updateR}
\tilde{\bm{R}}_l\leftarrow \bm{G}_{m,n}^H\tilde{\bm{R}}_l\bm{G}_{m,n},\ l=1,\dots,L.
\end{equation}

\textbf{end}

\State \textbf{end}

\noindent Denote $ \bm{V} $ by
\begin{equation}
\label{equ:Vdenotation}
\bm{V}=\prod_{m=1}^{L-1}\prod_{n=m+1}^L\bm{G}_{m,n}.
\end{equation}
\end{algorithmic}
{\bf Output:}\ The unitary matrix $ \bm{V} $.
\end{algorithm}

\bibliography{reference}
\bibliographystyle{IEEEtran}

\ifCLASSOPTIONcaptionsoff
  \newpag
\fi








\end{document}